\documentclass[final,leqno,onefignum,onetabnum]{siamltex}
\usepackage{amsfonts, graphicx, amsmath, amssymb,mathabx,float}
\usepackage[letterpaper, margin=1in]{geometry}
\usepackage{showkeys}
\usepackage[noadjust]{cite}

\newcommand{\vph}{\vphantom{A^{A}_{A}}}

\newcommand{\bu}{{\mathbf u}}
 \def\bn{\mathbf n}

  \def\bD{\mathbf D}

\newcommand{\bv}{\boldsymbol{v}}
\newcommand{\bfnu}{\boldsymbol{\nu}}

\newcommand{\partderiv}[2]{\ensuremath{\frac{\partial #1}{\partial #2}}}

\newcommand{\mb}[1]{\mathbf{#1}}

\newcommand{\cW}{\mathcal{W}}

\newcommand{\wof}{{\mathcal W}_{\text{\begin{tiny}{OF}\end{tiny}}}}
\newcommand{\dwof}{{\dot{\mathcal W}}_{\text{\begin{tiny}{OF}\end{tiny}}}}

   \newcommand{\tv}{T^{\mathrm{\begin{tiny}{v}\end{tiny}}}} \newcommand{\te}{T^{\text{\begin{tiny}{e}\end{tiny}}}}
     \newcommand{\gv}{{\mathbf g}^{\mathrm{\begin{tiny}{v}\end{tiny}}}}
      \newcommand{\gel}{{\mathbf g}^{\text{\begin{tiny}{e}\end{tiny}}}}

\renewcommand\div{\operatorname{div}}
\newcommand\curl{\operatorname{curl}}

\newcommand\tr{\operatorname{tr}}

\newcommand{\varperp}{\varepsilon_{\perp}}
\newcommand{\vara}{\varepsilon_a}

\def\tr{\mathop{\rm tr}\nolimits}
\def\curl{\mathop{\rm curl}\nolimits}

\newcommand{\wion}{\cW_\text{\begin{tiny}{ion}\end{tiny}}}
\newcommand{\wc}{\cW_\text{\begin{tiny}{elec}\end{tiny}}}

\newcommand{\bl}{{\mathbf l}}
\newtheorem{remark}{Remark}

\newcommand{\bE}{{\mathbf E}}

\newcommand{\bx}{{\mathbf x}}

\newcommand{\Ccchamber}{C_{\textrm{\begin{tiny}C\end{tiny}}}^{\textrm{\begin{
tiny}II\end{tiny}}}}

\newcommand{\Chcominus}{C_{{\textrm{\begin{tiny}H\end{tiny}}}{\textrm{\begin{
tiny}CO\end{tiny}}}_3^{\textrm{\begin{tiny}-\end{tiny}}}}}

\newcommand{\Jhcominus}{J_{{\textrm{\begin{tiny}H\end{tiny}}}_2{\textrm{\begin{
tiny}CO\end{tiny}}}_3^{-}}^{\textrm{\begin{tiny}II\end{tiny}}}}

\newcommand{\xx}{{\mathbf x}}
\newcommand{\bt}{{\mathbf t}}

\newcommand{\bnu}{{\boldsymbol\nu}}

\newcommand{\ff}{\boldsymbol f}

\newcommand{\calW}{\mathcal W}
\newcommand{\Pe}{\mathrm{Pe}}
\newcommand{\Rey}{\mathrm{Re}}
\newcommand{\FN}{\mathrm{F}}

\newcommand{\BN}{\mathrm{B}}

\title{Modeling of nematic electrolytes and nonlinear electroosmosis}

\author{M.~Carme Calderer\thanks{School of Mathematics, University of Minnesota, 
Minneapolis, MN 55455, USA.} \and Dmitry Golovaty \thanks{Department of Theoretical and Applied Mathematics, 
University of Akron, Akron, OH 44325, USA.} \and Oleg Lavrentovich \thanks{Liquid Crystal Institute, Kent State University, Kent, OH 44242} \and Noel J. Walkington
\thanks{Department of Mathematical Sciences, Carnegie-Mellon University, Pittsburgh, PA 15213}}

\begin{document}

\maketitle

\begin{abstract}
We derive a mathematical model of a nematic electrolyte based on the Leslie-Ericksen theory of liquid crystal flow. Our goal is to investigate the nonlinear electrokinetic effects that occur because the nematic matrix is anisotropic, in particular, transport of ions in a direction perpendicular to the electric field as well as quadratic dependence of the induced flow velocity on the electric field. The latter effect makes it possible to generate sustained flows in the nematic electrolyte that do not reverse their direction when the polarity of the applied electric field is reversed. From a practical perspective, this enables the design of AC-driven electrophoretic and electroosmotic devices. Our study of a special flow in a thin nematic film shows good qualitative agreement with laboratory experiments.
\end{abstract}

\begin{keywords}
Ericksen-Leslie, liquid crystals, electroosmosis, variational principles, asymptotics
\end{keywords}

\begin{AMS}
35Q35, 74E10, 49J40, 74A20, 80M30
\end{AMS}

\section{Introduction}

In this article, we derive equations governing the electrokinetics of a nematic electrolyte that consists of ions that diffuse and advect in the nematic liquid crystalline matrix.
The nematic electrolytes are characterized by the unique nonlinear phenomena that occur in these material systems due to anisotropy of conductivity and permittivity of the matrix.  

Electrokinetic phenomena are usually explored for binary systems in which an isotropic fluid with ions, called electrolyte, is in contact with a 
solid substrate or contains dispersed solid particles.  One distinguishes two sides of electrokinetics: electrophoresis, defined as motion of
particles dispersed in the electrolyte and electroosmosis, motion of an electrolyte with respect to the walls of a chamber. 
A necessary condition of electrokinetics is spatial separation of electric charges of opposite polarities (\cite{Bazant-Squires2010},
\cite{MorganGreen2003}  and \cite{re:russel89}).
In classic linear electrokinetics, charges are separated at the solid-electrolyte interface through chemical-physical processes such as
dissociation and selective adsorption and formation of permanent electric double layers \cite{re:russel89}.  An externally applied electric field imposes a
torque on the electric double layer and drives electrokinetic flows.  The driving force, proportional to the product of charge and field, 
is balanced by the viscous drag; the resulting velocities grow linearly with the electric field.  As a result, only a direct current (DC) 
field can be used to power linear electrokinetics, as an AC field would produce no net displacement. There is a growing interest in non-linear electrokinetics, in which the flow velocities grow as the square of the applied field.  Such a dependence allows one to use an AC field to drive stationary flows.  In the case of isotropic electrolytes, the corresponding effects are the so-called AC electrokinetics (ACEK) \cite{ramos1999ac} and induced-charge electrokinetics (ICEK) \cite{bazant2004induced,squires2006breaking}. The spatial charge is induced on energized electrodes (the case of ACEK), or on the "floating" polarizable (metal) particles located in an externally applied electric field (the case of ICEK). In both effects, the electrokinetic velocities grow as $E^2$, where one power of $E$ induces the charge near the highly polarizable metal surface, while the second power of $E$ drives these charges to trigger flows or to transport particles.  Both ACEK and ICEK combined with broken symmetry of electrodes or particles can lead to an AC-driven pumping of the fluids or electrophoresis of free particles \cite{ajdari2000pumping,bazant2004induced,squires2006breaking}.

 The studies of the liquid crystal-enabled electrokinetics are a part of a much larger field of liquid crystal colloids that is currently experiencing a great deal of interest partially as a result of the progress in the field of nanotechnology. Recent experiments, \cite{Hernandez_Navarro-Sagues2013,Hernandez_Navarro2015,LavrentovichLazoP2010,LazoLavrent2013,Lazo-Lavrent2014,Lavrent2014,Sasaki2014}, demonstrate that when the isotropic electrolyte is replaced with an anisotropic electrolyte, a liquid crystal containing ions, the electrokinetic flows become strongly nonlinear, with the velocities growing as a square of the electric field. For such a flow, if the polarity of the applied field is reversed, the direction of the flow remains unchanged, enabling alternating current (or AC) -driven electroosmosis and electrophoresis. The nonlinearity disappears as soon as the liquid crystal is melted into an isotropic phase. Despite the similarity in the quadratic field dependences of the flows, the mechanisms of the liquid crystal-enabled electrokinetics and the ACEK and ICEP effects in isotropic fluids are different, as discussed in \cite{lavrentovich2015active}.  Separation of charges and electrokinetics in isotropic electrolytes requires highly polarizable (metal) particles or interfaces; the isotropic electrolyte plays a supportive role, supplying the counterions. In the case of liquid crystal-enabled electrokinetics, the space charge is induced by the applied electric field at the distortions of the director field, thanks to the anisotropy of electric properties of the liquid crystal; no polarizable particles are needed to separate the charges and to generate the flows.
 
 Of a particular interest---from the point of view of this paper---are the experiments in \cite{Lavrent2014} where surface patterning of the plates bounding the nematic film is used to impose anchoring conditions on the nematic director. Surface patterning of liquid crystal cells thus allows one to impose a well-defined spatial variation of the director in the bulk \cite{Lavrent2014} and, consequently, the characteristics of the induced flow. This setup is especially amenable to theoretical analysis since the surface-induced director patterns in the available experiments are usually periodic, either one- or two-dimensional.

Our approach to modeling of nematic electrolytes follows the ideas that were originally used to obtain the Ericksen-Leslie equations of liquid crystalline flow.  The work of Leslie in this area spans a lifetime of research effort both to capture the appropriate physical phenomena and to provide a sound mathematical theory to describe the flow, with the main focus on understanding the balance of angular momentum. Leslie's thermodynamically self-consistent derivation appears to be the most straightforward method to formulate a consistent model for the nematic electrolytes. Although the exact form of the governing equations may be different if one was to follow a more rigorous procedure, we believe that we are able to capture the correct form of all principal contributions by imposing the proper thermodynamical structure, along with appropriate invariances and symmetries.

The variables of the model consist of the velocity field $\bv$ of the nematic, the pressure $p$ due to the incompressibility constraint, the unit director field $\bn$ representing the average molecular orientation at a given point, the electrostatic potential $\Phi$, and concentrations $\{c_k\}_{k=1}^N$, $N\geq 2$,  associated with the $N$ species of ions with valences $\{z_k\}_{k=1}^N$, respectively.  

Our development of the model follows the two main works by Leslie \cite{leslie1992continuum} and \cite{leslie79}. The former simplifies the previous approach by emphasizing the role of the rate of energy dissipation of the system and its direct connection with the viscous contributions to both the stress and the molecular force. These, combined with the variational insights on Leslie's works by Walkington \cite{Wa11}, allow for a more direct approach to the Ericksen-Leslie model.  As it was done by Leslie, we assume that the laws of balance of linear and angular momentum hold in local form. We supplement these by the local mass balances for the ions and the Maxwell's equations of electrostatics and postulate the equation of balance of energy at a time $t\geq 0$ for every subdomain $V$ of the domain $\Omega$ occupied by the nematic. The total energy of the system is now the sum of the Oseen-Frank free energy density, the entropic contribution due the presence of ions, and the
electrostatic energy.  The balance of energy involves the dissipation function that is required to be positive for all processes, according to the Second Law of Thermodynamics.  This condition yields the constitutive equations for the generalized stress tensor and the molecular force. 

In what follows, we assume that the diffusion and dielectric permittivity matrices of the system are uniaxial. The anisotropy of the diffusion matrix is fundamental in explaining the experimentally observed AC electroosmosis in a nematic film confined between the patterned plates. However, it is the combination of this anisotropy with the anisotropy of dielectric permittivity and the anisotropy of the viscosity that provides a rich variety of possible flow patterns \cite{LazoLavrent2013}. Note that, in the proposed model, we neglect dependence of the viscosity coefficients on concentration fields thus neglecting electro-rheologic effects. 
 
The theory of polyelectrolyte gels previously studied \cite{ChenCaldererMori2014}, \cite{CaldererChenMicekMori2013} provides a rigorous setting  to model electro-mechanic  interaction in liquid crystals. Given the dissipative character of these systems, our development also suggest that an alternate approach to generate governing equations and formulate boundary conditions may be based on the Onsager's principle for the Rayleighian functional.  

In the second part of the article, we illustrate the utility of our model by considering a particularly simple parameter regime in which the large system of governing equation essentially reduces to a single second order nonlinear ordinary differential equation. This regime qualitatively describes electroosmosis in a nematic liquid crystalline film constrained between two parallel plates. In \cite{Lavrent2014} an approach was developed to generate electrokinetic effects in a nematic electrolyte with surface-imposed distortions of molecular orientations induced by patterning of the plates. In the presence of the uniform electric field, these variations produce space charge separation that 
triggers electroosmotic flows in the liquid crystal. In particular, for the setup depicted in Fig. \ref{fig:exp}, the director orientation is periodically varying
in the vertical direction and an AC field is applied in the horizontal direction. It was observed that spatially periodic horizontal flow proceeds along the "guiding rails"
induced by molecular orientation with the direction of the flow independent of the sign of the field. 

In this paper we will only be interested in establishing qualitative similarity between the model and the experiment. To this end, we simplify the model as much as possible and consider a parameter regime that does not necessarily correspond to that in \cite{Lavrent2014} but it is still rich enough to result in the behavior similar to what was observed in \cite{Lavrent2014}. In particular, we assume that the liquid crystal domain is very thin and neglect the anisotropies of viscosity and dielectric permittivity, with the diffusion being the only source of anisotropy. Analysis and simulations of the simplified model shows features similar to experimental observations; in particular the quadratic dependence of the driving force on the electric field is recovered. Note that a detailed mathematical study of the experiment in \cite{Lavrent2014} using the methods described in this paper will appear elsewhere.

The following notational conventions are used throughout the rest of the paper. The trace of $A$ is given by $\tr{A}=\sum_iA_{ii}$ for any matrix $A\in M^{3\times3}$. The inner product of two matrices is defined as $A\cdot B=\mathrm{tr}\left(B^TA\right)$ for any $A,B\in M^{3\times3}$. These definitions immediately extend to second order tensors with components given by $3\times3$ matrices. The tensor product of ${\bf x,y}\in\mathbb R^3$ is the tensor ${\bf x}\otimes{\bf y}$ that assigns to each vector ${\bf c}\in\mathbb R^3$ the vector $({\bf y\cdot c}){\bf x}$. The divergence and the curl of a vector field ${\bf a}\in\mathbb R^3$ will be denoted by $\mathrm{div}\,{\bf a}$ and $\mathrm{curl}\,{\bf a}$, respectively. The divergence $\mathrm{div}\,S$ of a tensor field $S$ is a unique vector field with the property $(\mathrm{div}\,S)\cdot{\bf c}=\mathrm{div}\,(S^T{\bf c})$ for every constant ${\bf c}\in\mathbb R^3$. Here ${\bf x}\cdot{\bf y}$ is the Euclidean inner product of two vectors ${\bf x,y}\in\mathbb R^3$. Both the symbol $\frac{d}{dt}$ and the superimposed dot will be used interchangeably to represents the material time derivative.
 
  \section{Ericksen-Leslie Model for a Nematic Electrolyte}
  \subsection{Standard Nematic Model}
  We begin by reviewing Leslie's derivation of the classical Ericksen-Leslie model. Let $\Omega$, with the piecewise smooth boundary $\partial \Omega$, denote the domain occupied by the liquid crystal. Suppose that $\bv=\bv(\bx, t)$ and $\bn=\bn(\bx,t) $ 
  denote the velocity and director fields, respectively. The vector fields $\bf t(\bx,t)$ and $\bf l(\bx,t)$ represent 
 contact force and  contact couple per unit area of a surface element $S$, respectively. We assume that there exist the Cauchy stress tensor $T$  and the generalized stress tensor  $L$ (contact torque) such that 
 \begin{equation}
  {\bf t(\bx,t)}= T(\bx, t)\bnu(\bx,t), \quad {\bf l(\bx, t)}= L(\bx, t)\bnu(\bx, t). \label{Cauchy}
 \end{equation}
 We set
  \begin{equation}
   T= \te + \tv,  \quad L=L^{\mathrm{e}}+L^{\mathrm{v}} \label{total-stress}
  \end{equation}
 where $\bnu$ represents the unit outer normal to a surface $S$ at $\bx$, the tensor field $\te$ is the elastic stress, and $\tv$ is the anisotropic part of the viscous stress tensor. Following Leslie, we assume that there is no viscous torque tensor associated with $L$, that is, we set $L=L^{\mathrm{e}}$, where $L_e$ is the elastic torque tensor. 
  
We postulate the equations of balance of linear and angular momentum, together with the incompressibility assumption and the unit director field constraint:
  \begin{gather}
   \rho \dot\bv-\div T=\rho \boldsymbol f, \label{linear momentum0}\\
   \div\bv=0, \label{incompressibility0}\\
   \chi\ddot\bn+ \mathbf{g} -\div L=\rho\boldsymbol g, \label{angular momentum0}\\
   \bn\cdot\bn=1. \label{constraint0}
  \end{gather}
Here $\rho$ is the mass density and $\chi$ is the density of the moment of inertia of nematic rods. The symbols $\ff$ and $\boldsymbol g$ denote the density per unit mass of an applied external force and torque, respectively. The body torque $\mathbf g$ can be written as 
\begin{equation}
\label{eq:bt}
\mathbf g=\gv+\gel,
\end{equation}
where $\gel $ and $\gv$ denote the elastic and viscous contributions associated with director field rotations.

Below we assume that the nematic energy density is in the Oseen-Frank form
 \begin{eqnarray} \wof(\bn,\nabla\bn)= && \frac{1}{2}K_1(\div\bn)^2+\frac{1}{2}K_2(\bn\cdot\curl\bn)^2+\frac{1}{2}K_3|\bn\times\curl\bn|^2\nonumber\\
 +&&\frac{1}{2}(K_2+K_4)(\tr(\nabla\bn)^2-\tr^2(\nabla\bn)), \label{Oseen Frank}
  \end{eqnarray}
  where the Frank elastic constants $K_i,\ i=1,\ldots,4$ are assumed to satisfy the Ericksen's inequalities 
  \begin{equation}
   K_1>0, \, K_2>0, \, K_3>0, \, K_2\geq |K_4|, \, 2K_1\geq K_2+K_4, \label{Ericksen inequalities}
  \end{equation}
to guarantee existence of a global minimizer of the total energy
\begin{equation}
 U=\int_{\Omega}\wof(\bn, \nabla\bn)
   \end{equation}
   under appropriate boundary data \cite{V94}.
   
   \subsubsection{Dynamics}
    Among the different approaches that can be used to derive the constitutive equations for the fields $\te, \tv, L^{\mathrm{e}},  \gel,\gv$, we choose the line of reasoning proposed by Leslie that starts with postulating the balance laws (\ref{linear momentum0}) and (\ref{angular momentum0}) along with the equation of the energy balance. Following \cite{leslie79}, we let $R_{LC}$ be the rate of viscous dissipation per unit volume and assume that 
\begin{equation}
 \int_{V}\rho(\boldsymbol f\cdot\bv+ \boldsymbol g\cdot\dot\bn)+\int_{\partial V}(\bt\cdot\bv+\bl\cdot\dot\bn)=\frac{d}{dt} \int_{V}\left(\frac{1}{2}\rho|\bv|^2+ 
 \frac{1}{2}\chi|\dot\bn|^2+\wof\right)+ \int_{V} R_{LC}, \label{energy-balance}
 \end{equation}
for every subdomain $V\subseteq\Omega$ with the smooth boundary $\partial V$. The local form of (\ref{energy-balance})
\begin{equation}
 T\cdot\nabla\bv + L\cdot\nabla\dot\bn +\mathbf{g}\cdot\dot\bn=\dwof+R_{LC}, \label{energy-local}
 \end{equation}
follows via the divergence theorem from \eqref{Cauchy}, \eqref{linear momentum0}, and \eqref{angular momentum0}.
A simple computation shows that
\begin{equation}
 \nabla\dot\bn=(\nabla\bn)^{\cdot}+ \nabla\bn\nabla\bv,
 \end{equation}
hence
\begin{eqnarray}
\label{eq:213}
 \dwof=&&\frac{\partial \wof}{\partial \bn}\cdot\dot\bn+ \frac{\partial \wof}{\partial\nabla\bn}\cdot(\nabla\bn)^{\cdot}\nonumber\\
=&& \partderiv{\wof}{\bn}\cdot\dot\bn + \partderiv{\wof}{\nabla\bn}\cdot\left(\nabla\dot\bn-\nabla\bn\nabla\bv\right).\end{eqnarray}
Substituting \eqref{eq:213} into \eqref{energy-local} yields
 \[
 T\cdot\nabla\bv + L\cdot\nabla\dot\bn +\mathbf{g}\cdot\dot\bn= \partderiv{\wof}{\bn}\cdot\dot\bn + \partderiv{\wof}{\nabla\bn}\cdot\left(\nabla\dot\bn-\nabla\bn\nabla\bv\right)+R_{LC},
 \]
 so that \eqref{total-stress} and \eqref{eq:bt} give
  \begin{align}
 \left(\te + \tv+(\nabla\bn)^T\partderiv{\wof}{\nabla\bn}\right)\cdot\nabla\bv + \left(L^{\mathrm{e}}-\partderiv{\wof}{\nabla\bn}\right)\cdot\nabla\dot\bn \nonumber \\ +\left(\gel+\gv-\partderiv{\wof}{\bn}\right)\cdot\dot\bn=R_{LC}. \label{energy-local1}
 \end{align}
 The Second Law of Thermodynamics in the form of the Clausius-Duhem inequality, together with the appropriate smoothness assumptions,  implies the positivity of the rate of viscous
  dissipation  function
\begin{equation}  R_{LC}(\xx,t)\geq 0, \quad \forall \xx\in\Omega, \, t>0,\label{postive-dissipation}\end{equation}
for all dynamical processes $\{\bv, \dot\bn\}$.  Specifically,  given $\bn(\bx,t)$ and $\nabla\bn(\bx,t)$, the inequality (\ref{postive-dissipation}) must hold  for arbitrary choices at $(\bx, t)$ of $\bv,\ \nabla\bv,\ \dot\bn$ and $\nabla\dot\bn$.
This yields the  constitutive relations 
\begin{eqnarray}
 \te&=&-pI -(\nabla\bn)^T\partderiv{\wof}{\nabla\bn},{\label{elastic-constitutive-equations}} \\
 \gel&=&\partderiv{\wof}{\bn}+\lambda\bn, {\label{elastic-constitutive-equations-g}}\\
 L^{\mathrm{e}}&=&\partderiv{\wof}{\nabla\bn}, {\label{elastic-constitutive-equations-L}}
\end{eqnarray}
where $p$ and $\lambda$ are the Lagrange multiplier corresponding to the constraints \eqref{incompressibility0} and \eqref{constraint0}, respectively.

It also follows that 
\begin{equation}
R_{LC}=  \gv\cdot\dot\bn+\tv\cdot\nabla\bv. \label{dissipation-functional}
\end{equation}
The arguments in \cite{leslie1992continuum} then yield the total viscous stress
\begin{equation}
\tv=\alpha_1 \left(D(\bv)\bn\cdot\bn\right)\bn\otimes\bn+ \alpha_2 \ring{\bn}\otimes\bn + \alpha_3{\bn}\otimes\ring\bn+ \alpha_4D(\bv)+\alpha_5 D(\bv)\bn\otimes\bn+\alpha_6\bn\otimes D(\bv)\bn,  \label{viscous-stressL}
\end{equation}
and the viscous molecular force
\begin{equation}
 \gv=\gamma_1\ring{\bn} +\gamma_2 D(\bv)\bn, \label{viscous-molecular-forceL}
\end{equation}
where
\begin{equation}
\label{eq:gammas}
\gamma_1=\alpha_3-\alpha_2,\ \ \gamma_2=\alpha_6-\alpha_5,
\end{equation}
   and $\ring\bn=\dot\bn-W(\bv)\bn$ is the Lie derivative of $\bn$. Further, 
 \begin{eqnarray}
       D(\bv)=\frac{1}{2}(\nabla\bv+ \nabla\bv^T)\,\,\, \textrm{and}\, \, \, W(\bv)= \frac{1}{2}(\nabla\bv-\nabla\bv^T)
         \end{eqnarray}
represent the symmetric and skew parts of the velocity gradient $\nabla\bv$, respectively.
Note that an even more general expression \cite{kleman-soft} that involves the gradient of $\ring\bn$ can be established for
\eqref{viscous-stressL}-\eqref{viscous-molecular-forceL}, although we chose not to include terms of this type here. 

In what follows, we assume that the Parodi's relation \cite{leslie1992continuum} given by
\begin{equation}
\alpha_6-\alpha_5=\alpha_2+ \alpha_3. \label{parodi}
\end{equation}
 holds. This relation is necessary to ensure the variational structure of the system of equations and thus the equivalency of the equation of balance of linear momentum \eqref{linear momentum0} to that derived via the Onsager's principle. Then \eqref{dissipation-functional} and \eqref{viscous-stressL}-\eqref{viscous-molecular-forceL} give
\begin{eqnarray}
2R_{LC}= \alpha_1(\bn\cdot D(\bv)\bn)^2 + 2\gamma_2(\ring\bn\cdot D(\bv)\bn)+\alpha_4 |D(\bv)|^2+ (\alpha_5+\alpha_6)|D(\bv)\bn|^2+\gamma_1|\ring\bn|^2. \label{RLC2}
\end{eqnarray}
Ericksen \cite{Er87} gave sufficient conditions for the positivity of $R_{LC}$ in the following
\begin{proposition} Suppose that (\ref{parodi}) and the inequalities
\begin{equation}
\alpha_4>0,\ \ \ \alpha_1+\frac{3}{2}\alpha_4+\alpha_5+\alpha_6>0,\ \ \ \gamma_1>0,\ \ \ \gamma_1(2\alpha_4+\alpha_5+\alpha_6)\geq {\gamma_2^2}. \label{leslie-ineq}
\end{equation}
hold. Then 
\begin{equation}
 R_{LC}\geq 0
\end{equation}
and
\begin{equation}
\tv=\partderiv{R_{LC}}{\nabla\bv}\mbox{  \textrm{and}  } {\gv}=\partderiv{R_{LC}}{\dot\bn}.
\end{equation}
Moreover  $ R_{LC}\equiv 0$ if and only if $\nabla\bv={\bf 0}$ and $\ring\bn={\bf 0}$.
\end{proposition}

\subsubsection{Boundary Conditions}

Since the energy law \eqref{energy-balance} for the Ericksen-Leslie system must hold in the entire domain $\Omega$,
it follows from the equations \eqref{linear momentum0}-\eqref{constraint0} and \eqref{elastic-constitutive-equations}-\eqref{dissipation-functional} that
\begin{eqnarray}
  \int_{\partial\Omega}\left\{T\nu\cdot\bv+L\nu\cdot\dot\bn\right\} 
= \int_{\partial\Omega}\left\{\boldsymbol t\cdot\bv+\bl\cdot\dot\bn\right\}.
 \end{eqnarray}
This equation should be valid for all dynamical processes $\{\bv, \dot\bn\}$ therefore the boundary conditions on $\partial \Omega$ should be of the form
\begin{eqnarray}
 && T\bnu=\hat{\boldsymbol t} \quad {\textrm{or}}\quad \bv =\boldsymbol 0, \quad {\textrm{and}} \\
&&  L\bnu=\boldsymbol{\hat {\boldsymbol l}}\quad {\textrm{or}}\quad \bn=\hat\bn,
\end{eqnarray} 
where $\hat{\boldsymbol t}$, $\hat{\boldsymbol l}$ and $\hat{\boldsymbol{\bn}}$ are prescribed vector fields on $\partial\Omega$ with 
$|\hat{\bn}|=1$ and 
$T$ and $L$ given by \eqref{total-stress}, \eqref{elastic-constitutive-equations}, \eqref{viscous-stressL} and \eqref{elastic-constitutive-equations-L}. Observe that the fields $\hat{\boldsymbol t}$ and $\hat{\boldsymbol l}$ can be time-dependent in this formulation.

\subsection{Nematic Electrolyte}
Suppose now that the domain $\Omega\subset{\mathbf R}^3$ is occupied by a nematic electrolyte that contains ions.
Later on, we will assume that some parts of $\partial\Omega$ correspond to conducting electrodes on which we will prescribe values of the electrostatic potential $\Phi$, 
while the other parts of the boundary will be assumed to be electrically insulated. In this section, however, we will impose time-independent, Dirichlet boundary data on 
the potential $\Phi$ everywhere on $\partial\Omega$, corresponding to nematic being surrounded by conductors held at fixed potentials. 
This problem setup is chosen for simplicity because we do not expect the boundary data on the electric field to affect the constitutive expressions on electrostatic
forces in the bulk of the nematic electrolyte. 

Suppose that there are $N>1$ families of charged ions present in the liquid crystal at concentrations $c_k$, with  valences $z_k$, where $1\leq k\leq N$.  Let the  velocity fields of the ions be denoted by $\{\bu_k\}_{1\leq k\leq N}$. In what follows, we assume that the system is in the dilute regime so that the particles are not subject to mutual interaction. The continuity equations for the ions are given by
\begin{eqnarray}
 \frac{\partial c_k}{\partial t}+\div(c_k\mb{u}_k)=0\quad\text{ in }\Omega, \, k=1, ...N.\label{ck}
 \end{eqnarray}

Motivated by standard results of the theory of isotropic diffusion, we assume that the rate of dissipation associated with the mobility of ions in the nematic is a 
quadratic function of the relative velocity of the ions with respect to the liquid crystalline medium. We set
\begin{eqnarray}
R= R_{LC}+ \sum_{k=0}^{N}k_B\theta c_k \mathcal D_k^{-1}(\bu_k-\bv)\cdot(\bu_k-\bv), \label{D2} 
\end{eqnarray}
where $R_{LC}$ is given by \eqref{RLC2} the {\it diffusion matrix} $\mathcal D_k$ is anisotropic, reflecting the fact that the mobilities of the $k$-th species in the directions parallel and perpendicular to the nematic director are generally different. The parameter $k_B$ in \eqref{D2} is the Boltzmann constant and $\theta$ is the absolute
temperature \cite{kleman-soft}. The ions also contribute to the free energy of the system via an entropic energy density term 
\begin{equation}
\label{enion}
  \cW_\text{ion}=k_B\theta\sum_{k=1}^{N}c_k\ln c_k.
\end{equation}

The electric displacement vector $\mathbf D$ of the nematic liquid crystal is given by
\begin{equation}
 \label{D-E}
\bD= \varepsilon_0\varepsilon\bE,
\end{equation} 
where $\varepsilon=I+\chi$ is the dielectric permittivity matrix. Letting $\varepsilon_{\|}$ and $\varperp$ to
represent the dielectric permittivities when $\bE$ is parallel and
perpendicular to $\bn$, respectively, and denoting
$\vara=\varepsilon_{\|}-\varperp$, we have that
\begin{equation}
\varepsilon(\bn)= \varperp I + \vara \bn\otimes\bn.\label{effective-dielectric-matrix}
\end{equation}

The fields $\bE$ and $\bD$ satisfy the Maxwell's equations of electrostatics
\begin{equation}
\label{eq:max}
\bE=-\nabla\Phi,
\qquad 
\div{\bD} = \sum_{k=1}^N q z_k c_k.
% =\nabla\cdot{(\varepsilon_0\varepsilon(\bn)\nabla\Phi)}
% =-\rho_e, \mbox{ where } \rho_e:=\sum_{k=1}^N qc_kz_k.
\end{equation}
that hold in $\Omega$, subject to time-independent Dirichlet boundary
data for $\Phi$ on $\partial\Omega$. Here the parameter $q$ denotes the elementary charge. The electrostatic energy
\begin{equation}
 \label{enelec}
\wc = -\frac{1}{2}\bD\cdot\bE
+\sum_{k=1}^N q z_k c_k \Phi
= -\frac{\varepsilon_0}{2} \varepsilon(\bn) \nabla\Phi\cdot \nabla\Phi
+\sum_{k=1}^N q z_k c_k \Phi,
\end{equation}
is clearly {\em nonlocal} because $\Phi$ is determined by solving the
second equation in \eqref{eq:max} for the given $\bn,\ c_k$, and
the appropriate boundary data on $\Phi$. The arguments of Leslie in
the purely mechanical case \cite{leslie1992continuum} rely on
formulating a {\em local} energy balance \eqref{energy-balance} for a
material control volume $V\subset\Omega$ with the balance being
assumed to hold for any dynamical process $\{\bv, \dot\bn\}$ with a
support in $V$. The following simple proposition allows for
localization of the time derivative of the total electrostatic energy.

\begin{proposition}
  Suppose that $\Phi$ satisfies the second equation in \eqref{eq:max},
  subject to time-independent Dirichlet boundary data on
  $\partial\Omega$. If the support of $\{\bv,
  \bu_1,\ldots,\bu_N,\dot\bn\}$ is contained in $V$, that is the rates
  $\bv, \bu_1,\ldots,\bu_N,\dot\bn$ all vanish in
  $\Omega\backslash\bar V$, then
\begin{eqnarray}
\frac{d}{dt}\int_\Omega\wc
&=&\int_V\left\{
-\frac{\varepsilon_0}{2}\frac{d\varepsilon(\bn)}{dt}\nabla\Phi\cdot\nabla\Phi
- \sum_{k=1}^N q z_k c_k \nabla\Phi \cdot \left(\bv-\bu_k\right)
+\varepsilon_0\varepsilon(\bn)\nabla\Phi
  \cdot\left(\nabla\bv^T\nabla\Phi\right)\right\} \nonumber \\
&& + \int_{\partial V} 
\sum_{k=1}^N q z_k c_k \Phi \left(\bv-\bu_k\right).\bfnu.
\label{eq:elloc}
\end{eqnarray}
\end{proposition}
\begin{proof}
  Using \eqref{ck}, \eqref{enelec}, integration by parts,
  incompressibility of the flow, our assumptions on the dynamical
  process, and writing $\rho_e = \sum_{k=1}^N q z_k c_k$, we have
\begin{eqnarray*}
\lefteqn{
\frac{d}{dt}\int_\Omega\wc 
= \frac{d}{dt}\int_\Omega\left\{-\frac{1}{2}(\varepsilon_0\varepsilon(\bn)\nabla\Phi\cdot \nabla\Phi)+\rho_e\Phi\right\} } \\ 
&=&\int_\Omega\left\{-\frac{\varepsilon_0}{2}\frac{d\varepsilon(\bn)}{dt}\nabla\Phi\cdot\nabla\Phi+\left\{\frac{\partial\rho_e}{\partial t}+\div\left(\rho_e\bv\right)\right\}\Phi\right\} +\int_\Omega\left\{-\varepsilon_0\varepsilon(\bn)\nabla\Phi\cdot\frac{d}{dt}\left(\nabla\Phi\right)+\rho_e\frac{d\Phi}{dt}\right\} \\ 
&=&\int_V\left\{-\frac{\varepsilon_0}{2}\frac{d\varepsilon(\bn)}{dt}\nabla\Phi\cdot\nabla\Phi+q\Phi\sum_{k=1}^N z_k\div \left(c_k\left(\bv-\bu_k\right)\right)\right\} \\ 
&& +\int_\Omega\left\{-\varepsilon_0\varepsilon(\bn)\nabla\Phi\cdot\left\{\nabla\left(\frac{d\Phi}{dt}\right)-\nabla\bv^T\nabla\Phi\right\}+\rho_e\frac{d\Phi}{dt}\right\} \\ 
&=&\int_V\left\{-\frac{\varepsilon_0}{2}\frac{d\varepsilon(\bn)}{dt}\nabla\Phi\cdot\nabla\Phi+q\Phi\sum_{k=1}^N z_k\div\left(c_k\left(\bv-\bu_k\right)\right)+\varepsilon_0\varepsilon(\bn)\nabla\Phi\cdot\left(\nabla\bv^T\nabla\Phi\right)\right\} \\ 
&& +\int_\Omega\frac{d\Phi}{dt}\left\{\div{(\varepsilon_0\varepsilon(\bn)\nabla\Phi)}+\rho_e\right\}-\int_{\partial\Omega}\frac{\partial\Phi}{\partial t}\left\{\varepsilon_0\varepsilon(\bn)\nabla\Phi\cdot\bfnu\right\},
\end{eqnarray*}
and \eqref{eq:elloc} follows from \eqref{eq:max} and the fact that the potential does not depend on time on $\partial\Omega$.
\end{proof}

\begin{remark}
  Note that by using \eqref{effective-dielectric-matrix} and the first
  equation in \eqref{eq:max}, the first integrand on the right hand
  side of \eqref{eq:elloc} can be written as
\[\frac{\varepsilon_0}{2}\frac{d\varepsilon(\bn)}{dt}\nabla\Phi\cdot\nabla\Phi=\frac{\varepsilon_0\varepsilon_a}{2}\frac{d}{dt}\left(\bn\otimes\bn\right)\nabla\Phi\cdot\nabla\Phi=\varepsilon_0\varepsilon_a\left(\bn\cdot\nabla\Phi\right)\left(\dot\bn\cdot\nabla\Phi\right)=\varepsilon_0\varepsilon_a\left(\bE\otimes\bE\right)\bn\cdot\dot\bn.\]
Further, recalling the definition \eqref{D-E} of $\bD$ and using the
first equation in \eqref{eq:max}, the third integrand on the right
hand side of \eqref{eq:elloc} can be written as
\[\varepsilon_0\varepsilon(\bn)\nabla\Phi\cdot\left(\nabla\bv^T\nabla\Phi\right)=\bD\cdot\nabla\bv^T\bE=\left(\bE\otimes\bD\right)\cdot\nabla\bv,\]
so that \eqref{eq:elloc} takes the form
\begin{eqnarray}
\label{eq:elloc1}
\frac{d}{dt}\int_\Omega\wc
&=& \int_V\left\{
-\varepsilon_0\varepsilon_a\left(\bE\otimes\bE\right)\bn\cdot\dot\bn
+\left(\bE\otimes\bD\right)\cdot\nabla\bv
-\sum_{k=1}^N q z_k c_k \nabla \Phi \cdot \left(\bv-\bu_k\right)
\right\} \nonumber \\
&& + \int_{\partial V} 
\sum_{k=1}^N q z_k c_k \Phi \left(\bv-\bu_k\right)\cdot\bfnu
\end{eqnarray}
\end{remark}
The material time derivative of $\int_V\cW_\text{ion}$ will also enter the
energy balance.
\begin{proposition} Suppose that (\ref{ck}) holds. Then
\begin{equation}
  \frac{d}{dt}\int_{V}\wion
  = - \int_V\sum_{k=1}^{N}
  k_B \theta \nabla c_k \cdot \left(\bv-\bu_k\right)
  + \int_{\partial V} \sum_{k=1}^{N} k_B\theta\left(\ln c_k+1\right)
  c_k\left(\bv-\bu_k\right)\cdot\bfnu.
  \label{derivative-wion}
\end{equation}
\end{proposition}
\begin{proof}
  Using \eqref{enion} for $\wion$ and the mass balances \eqref{ck}
  gives
  \begin{eqnarray*}
 \frac{d}{dt}\int_V\wion
 &=&\int_V\frac{d}{dt}\left\{k_B\theta\sum_{k=1}^{N}c_k\ln c_k\right\}
 =\int_V\sum_{k=1}^{N}\left\{k_B\theta\left(\ln c_k+1\right)
   \left(\frac{\partial c_k}{\partial t}
     +\div\left(c_k\bv\right)\right)\right\} \\ 
 &=&\int_V \sum_{k=1}^{N}
 k_B\theta\left(\ln c_k+1\right)
 \div\left(c_k\left(\bv-\bu_k\right)\right) \\
 &=&  -\int_V\sum_{k=1}^{N}
 k_B \theta \nabla c_k \cdot \left(\bv-\bu_k\right)
 + \int_{\partial V} \sum_{k=1}^{N} k_B\theta\left(\ln c_k+1\right)
   c_k\left(\bv-\bu_k\right)\cdot\bfnu.
 \end{eqnarray*}
\end{proof}

Combining \eqref{eq:elloc1} and \eqref{derivative-wion} we obtain
\begin{eqnarray}
\frac{d}{dt}\left\{\int_{V}\wion+\int_\Omega\wc\right\} 
&=& \int_V\left\{\sum_{k=1}^Nc_k\nabla \mu_k\cdot\left(\bu_k-\bv\right)-\varepsilon_0\varepsilon_a\left(\bE\otimes\bE\right)\bn\cdot\dot\bn+\left(\bE\otimes\bD\right)\cdot\nabla\bv\right\} \nonumber \\
&& - \int_{\partial V}\sum_{k=1}^Nc_k \mu_k\left(\bu_k-\bv\right)\cdot\boldsymbol\nu,
\label{eq:ioel}
\end{eqnarray}
where the quantities
\begin{equation}
  \label{eq:elcp}
  \mu_k
  = k_B \theta (\ln(c_k)+1) + qz_k\Phi
  = \frac{\partial}{\partial c_k} \Big(\wion + q z_k c_k \Phi \Big),
  \qquad k=1,\ldots,N,
\end{equation}
are the {\em electrochemical potentials} of the ions.

In order to establish the set of governing equations, we now extend
the procedure carried out above in the purely mechanical case. We
postulate the same local forms of the balance of both linear
\eqref{linear momentum0} and angular \eqref{angular momentum0}
momenta, coupled with the constraint relations
\eqref{incompressibility0} and \eqref{constraint0} and assume that the
mass balances \eqref{ck} hold in $\Omega$ along with the Maxwell's
equations of electrostatics \eqref{eq:max} that hold in $\mathbb R^3$.
The equation of balance of energy in an arbitrary subdomain
$V\subset\Omega$ for any isothermal dynamical process \[\{\bv,
\bu_1,\ldots,\bu_N,\dot\bn\},\] with a support in $V$ is then given by
\begin{multline} 
 \int_{V}\rho(\boldsymbol f\cdot\bv + \boldsymbol g\cdot \dot\bn)
 + \int_{\partial V} \left(\bt\cdot\bv + \bl\cdot\dot\bn
   -\sum^{N}_{k=1}c_k \mu_k(\bu_k-\bv)
   \cdot\boldsymbol \nu \right) \\
   =\frac{d}{dt} \left\{\int_{V}\left(\frac{1}{2}\rho|\bv|^2+ \frac{1}{2}\chi|\dot\bn|^2+
\cW_\text{ion}+\wof\right)+\int_\Omega\wc\right\}+ \int_{V} R.
\label{ionic-energy-balance}
\end{multline}
The
additional boundary term $\sum^{N}_{k=1}c_k
\mu_k(\bu_k-\bv)\cdot\boldsymbol \nu$, which does not appear in
\eqref{energy-balance}, represents the energy transported across the
boundary by the ions.
Equations \eqref{energy-local1}, \eqref{D2}, and \eqref{eq:ioel} allow to express \eqref{ionic-energy-balance} in the local form
  \begin{align}
 \left(\te + \tv+(\nabla\bn)^T\partderiv{\wof}{\nabla\bn}-\bE\otimes\bD\right)\cdot\nabla\bv + \left(L^{\mathrm{e}}-\partderiv{\wof}{\nabla\bn}\right)\cdot\nabla\dot\bn \nonumber \\ +\left(\gel+\gv-\partderiv{\wof}{\bn}+\varepsilon_0\varepsilon_a\left(\bE\otimes\bE\right)\bn\right)\cdot\dot\bn-\sum_{k=1}^Nc_k\nabla\mu_k\cdot\left(\bu_k-\bv\right) \\ =R_{LC}+\sum_{k=0}^{N}\left\{k_B\theta c_k \mathcal D_k^{-1}(\bu_k-\bv)\cdot(\bu_k-\bv)\right\} \nonumber. \label{energy-local-ions}
 \end{align}
 The necessary conditions for positivity of the dissipation functional
 $R$ required by the Clausius-Duhem inequality for an arbitrary
 admissible dynamical process then give the following analogs of the
 constitutive relations
 \eqref{elastic-constitutive-equations}-\eqref{elastic-constitutive-equations-L}
 which account for the presence of ions and the electric field
\begin{eqnarray}
 \te&=&-pI -(\nabla\bn)^T\partderiv{\wof}{\nabla\bn}+\bE\otimes\bD,{\label{elastic-constitutive-equations-ions}} \\
 \gel&=&\partderiv{\wof}{\bn}
 -\varepsilon_0\varepsilon_a (\bE \otimes \bE)\bn +\lambda \bn, 
 {\label{elastic-constitutive-equations-g-ions}}\\
 L^{\mathrm{e}}&=&\partderiv{\wof}{\nabla\bn}, {\label{elastic-constitutive-equations-L-ions}} \\
\bu_k &=& \bv-\frac{1}{k_B\theta}\mathcal{D}_k\nabla \mu_k,
\qquad  k=1,\ldots,N,
\label{eq:chconst}
\end{eqnarray}
along with the relations \eqref{viscous-stressL} for the viscous
stress $\tv$ and \eqref{viscous-molecular-forceL} for the viscous
molecular force $\gv$.  As in the clasical Ericksen-Leslie system, the
fields $p$ and $\lambda$ in
\eqref{elastic-constitutive-equations-ions}-\eqref{elastic-constitutive-equations-g-ions},
are the Lagrange multipliers corresponding to the constraints
\eqref{incompressibility0} and \eqref{constraint0}, respectively.  We
are now ready to formulate the full set of equations governing the
evolution of a nematic electrolyte.
\begin{proposition}
  Suppose that the continuity equations \eqref{ck}, the linear
  momentum balance \eqref{linear momentum0}, the angular momentum
  balance \eqref{angular momentum0}, and the Maxwell's equations of
  electrostatics \eqref{eq:max} hold in $\Omega$. Further, suppose
  that the energy balance \eqref{ionic-energy-balance} holds in every
  subdomain $V\subset\Omega$ for any dynamical process \[\{\bv,
  \bu_1,\ldots,\bu_N,\dot\bn\},\] with a support in $V$. Then the
  necessary conditions for positivity of the dissipation functional
  $R$ in (\ref{D2}) are
  \eqref{viscous-stressL}-\eqref{viscous-molecular-forceL} and
  \eqref{elastic-constitutive-equations-ions}-\eqref{eq:chconst}.
  
  The system of equations for the nematic electrolyte is as follows
  \begin{gather}
   \partderiv{c_k}{t}+\div\left(c_k\left[\bv-\frac{1}{k_B\theta}\mathcal D_k\nabla \mu_k\right]\right)=0, \label{anisotropic-NP1} \\
   -\div(\varepsilon_0\varepsilon(\bn)\nabla \Phi)=
\sum^{N}_{k=1} qz_kc_k, \label{ansiotropic-poissson2}\\
   \rho \dot\bv-\div\left(-pI - (\nabla\bn)^T\partderiv{\wof}{\nabla\bn}+\varepsilon_0\left(\nabla\Phi\otimes\nabla\Phi\right)\varepsilon(\bn)+\tv\right)=\rho \boldsymbol f,
   \label{linear momentum}\\
   \div\bv=0, \label{incompressibility}\\
   \chi \ddot\bn+\partderiv{\wof}{\bn}-\varepsilon_0\varepsilon_a\left(\nabla\Phi\otimes\nabla\Phi\right)\bn+\div\left(\partderiv{\wof}{\nabla\bn}\right)+ \gv+\lambda\bn=\rho\boldsymbol g, \label{angular momentum}\\
   \bn\cdot\bn=1, \label{constraint}
  \end{gather}
  where $\tv$ is given by \eqref{viscous-stressL} and $\gv$ is given by \eqref{viscous-molecular-forceL}. 
\end{proposition}

\begin{proof}
The equation \eqref{anisotropic-NP1} immediately follows from \eqref{eq:chconst} and \eqref{ck}. The remaining equations follow by substituting the expressions for the appropriate stresses, using the definitions of $\bD$ and $\bE$, and by recalling that ${\bf a}\otimes A{\bf b}=({\bf a\otimes b})A^T$ for any ${\bf a,b} \in \mathbb R^3$ and $A\in M^{3\times3}$.
\end{proof}

\begin{remark}
From the equation (\ref{anisotropic-NP1}), we identify the tensor
\begin{equation}
\mathcal M_k=  \frac{qz_k}{k_B\theta} \mathcal D_k,
\end{equation}
as the {\it mobility tensor} of the $k$th species.  Note that, when $\mathcal D_k$ is a multiple of the identity,  $\mathcal M$ is the analog of  Einstein's mobility relation of electrons in a gas ($z_k=-1$). The conductivity matrix (mobility times charge density) of the $k$-th species is now given by
\begin{equation}
\sigma^k=\frac{1}{k_B \theta} c_kz_k^2 q^2 \mathcal D_k. \label{conductivity}
\end{equation}
Note that the corresponding resistivity matrix is equal to the inverse of $\sigma^k.$
\end{remark}

\subsubsection{Boundary Conditions}

Since the energy law \eqref{ionic-energy-balance} for the nematic electrolyte has to hold in the entire domain $\Omega$,
it follows from the equations \eqref{anisotropic-NP1}-\eqref{constraint} that
\begin{eqnarray}
  \int_{\partial\Omega}\left\{T\nu\cdot\bv+L\nu\cdot\dot\bn-\sum^{N}_{k=1}c_k \mu_k(\bu_k-\bv)
    \cdot\boldsymbol \nu\right\} = \int_{\partial\Omega}\left\{\boldsymbol t\cdot\bv+\bl\cdot\dot\bn-\sum^{N}_{k=1}j_k\right\},
\end{eqnarray}
where $j_k,\ k=1,\ldots,N$ represents the normal energy flux
associated with the transport of the $k$-th species of ions across the
boundary.  This equation should be valid for all dynamical processes
$\{\bv, \bu_1,\ldots,\bu_k,\dot\bn\}$ therefore the boundary
conditions on $\partial \Omega$ should be of the form
\begin{eqnarray}
  T\bnu=\hat{\boldsymbol t} 
  &\quad {\textrm{or}}\quad&
  \bv =\boldsymbol 0, \qquad {\textrm{and}} \\
  L\bnu=\boldsymbol{\hat {\boldsymbol l}}
  & {\textrm{or}}& 
  \bn=\hat\bn, \qquad {\textrm{and}} \\
  c_k \mu_k(\bu_k-\bv)\cdot\boldsymbol\nu=\hat j_k
  &{\textrm{or}} & \bu_k\cdot\boldsymbol\nu=\bv\cdot\boldsymbol\nu,
  \qquad k=1,\ldots,N,
\end{eqnarray} 
where $\hat{\boldsymbol t}$, $\hat{\boldsymbol l}$, $\hat{\boldsymbol{\bn}}$, and $\hat j_k,\ k=1,\ldots,N$ are prescribed fields on $\partial\Omega$ with 
$|\hat{\bn}|=1$ and $T$ and $L$ given by \eqref{total-stress}, \eqref{elastic-constitutive-equations-ions}, \eqref{viscous-stressL} and \eqref{elastic-constitutive-equations-L-ions}. Observe that the fields $\hat{\boldsymbol t}$, $\hat{\boldsymbol l}$, $\hat j_k,\ k=1,\ldots,N$ can be time-dependent in this formulation.

The set of the boundary conditions should be supplemented by the
boundary data on the electric field. Here, we will impose the
Dirichlet conditions on the poential on the boundary between the
nematic electrolyte and a conductor
\begin{equation}
\label{eq:con-el}
\Phi|_{\partial\Omega}=\Phi_0,
\end{equation}
for some prescribed $\Phi_0$. On the boundary between the nematic and
an insulating medium \cite{ll84}, we will impose the condition of the
zero jump of the normal component of the displacement $\bD$, that is
\begin{equation}
\left[\bD\cdot\nu\right]_{\partial\Omega}=0,
\end{equation}
where $\left[\cdot\right]_{\partial\Omega}$ represent the jump of a quantity in the brackets across $\partial\Omega$. In this case, the equations of electrostatics have to be
solved in $\mathbf{R}^3$.

\subsubsection{Variational Structure}
Setting 
\begin{eqnarray*}
  \calW(\bn,\nabla \bn, \Phi, \nabla \Phi, c_1,\ldots, c_N)
  &=& W_{OF}(\bn, \nabla \bn) 
  + \calW_{ion}(c_1, \ldots, c_N) 
  + \calW_{elec}(\phi,\nabla \phi, c_1, \ldots, c_N) \\
  &=& W_{OF}(\bn, \nabla \bn) 
  + \sum_{k=1}^N \left( k_B \theta c_k \ln(c_k) + q z_k c_k \Phi \vph\right)
  - \frac{\varepsilon_0}{2} \varepsilon(\bn) \nabla\Phi\cdot\nabla\Phi,
\end{eqnarray*}
Maxwell's equation \eqref{anisotropic-NP1} and the balance laws
\eqref{ansiotropic-poissson2}, \eqref{linear momentum} and
\eqref{angular momentum} may be written as
\begin{gather*}
  \partderiv{c_k}{t}
  + \div\left(c_k (\bv - 
    (1/k_B\theta) \nabla \mu_k) \vph \right)=0, \\
  \div\left( \partderiv{\calW}{\nabla \Phi} \right) 
  = \partderiv{\calW}{\Phi}, \\
  \rho \dot\bv
  -\div\left(-pI + \partderiv{R}{\nabla\bv}
    - (\nabla\bn)^T \partderiv{\calW}{\nabla\bn}
    - \nabla\Phi \otimes \partderiv{\calW}{\nabla\Phi}
  \right)=\rho \boldsymbol f, \\
  \chi \ddot\bn
  + \partderiv{R}{\dot{\bn}}
  - \div \left(\partderiv{\calW}{\nabla\bn}\right)
  + \partderiv{\calW}{\bn}
  +\lambda\bn=\rho\boldsymbol g.
\end{gather*}
Introducing the Legendre transform of $\calW$,
$$
\calW^*(\bn,\nabla \bn, \Phi, \nabla \Phi, \mu_1,\ldots, \mu_N)
= -W_{OF}(\bn, \nabla \bn) 
+ k_B \theta \sum_{k=1}^N \exp\left((\mu_k - qz_k \Phi)/k_B\theta -1\vph\right)
+ \frac{\varepsilon_0}{2} \varepsilon(\bn) \nabla\Phi\cdot\nabla\Phi,
$$
the dual relations take the form
$$
c_k = \partderiv{\calW^*}{\mu_k} 
= \exp\left((\mu_k - qz_k \Phi)/k_B\theta -1\vph\right)
\qquad \text{ and } \qquad
\mu_k = \partderiv{\calW}{c_k} = k_B \theta (\ln(c_k)+1) + q z_k \Phi.
$$
Computing $\nabla \calW^*$ and rearranging the terms shows
\begin{eqnarray*}
  \lefteqn{ \div \left( -\calW^* I 
      - (\nabla\bn)^T \partderiv{\calW}{\nabla\bn}
      - \nabla\Phi \otimes \partderiv{\calW}{\nabla\Phi} \right)
    = \div \left( -\calW^* I 
      + (\nabla\bn)^T \partderiv{\calW^*}{\nabla\bn}
      + \nabla\Phi \otimes \partderiv{\calW^*}{\nabla\Phi} \right) } \\
  &=& (\nabla \bn)^T \left( 
    \div \left(\partderiv{\calW^*}{\nabla\bn}\right)
    - \partderiv{\calW^*}{\bn} \right)
  + \left(\div \left(\partderiv{\calW^*}{\nabla \Phi}\right)
    - \partderiv{\calW^*}{\Phi} \right) \nabla \Phi
  - \sum_{k=1}^N \partderiv{\calW^*}{\mu_k} \nabla \mu_k \\
  &=& -(\nabla \bn)^T \left( \div \left(\partderiv{\calW}{\nabla\bn}\right)
    - \partderiv{\calW}{\bn} \right)
  - \left(\div \left(\partderiv{\calW}{\nabla \Phi}\right)
    - \partderiv{\calW}{\Phi} \right) \nabla \Phi
  - \sum_{k=1}^N c_k \nabla \mu_k,
\end{eqnarray*}
so the linear momentum equation can be written as
\begin{equation} \label{altLinearMomentum}
\rho \dot\bv
-\div\left(-(p+\calW^*) I + \partderiv{R}{\nabla\bv} \right)
+ (\nabla \bn)^T \left( \div \left(\partderiv{\calW}{\nabla\bn}\right)
  - \partderiv{\calW}{\bn} \right)
+ \sum_{k=1}^N c_k \nabla \mu_k 
=\rho \boldsymbol f.
\end{equation}
The energy estimate for the coupled system now follows upon
multiplying the equations for the concentrations by $\mu_k$, Maxwells
equation by $\Phi_t$, and the linear and angular momentum equations
by $\bv$ and $\dot{\bn}$ respectivly. Granted appropriate boundary
data, this gives
$$
\frac{d}{dt} \int_\Omega \left\{
(1/2) \left(\rho |\bv|^2 + \chi |\dot{\bn}|^2 \right) + \calW \vph\right\}
+ \int_\Omega \left\{
\partderiv{R}{\nabla \dot{\bn}} \cdot \dot{\bn}
+ \partderiv{R}{\nabla \bv} \cdot \nabla \bv 
+ \sum_{k=1}^N \frac{c_k}{k_B \theta} |\nabla \mu_k|^2
\vph\right\}
= \int_\Omega \rho \left(\boldsymbol f \cdot\bv 
+ \boldsymbol g\cdot \dot{\bn} \vph\right).
$$
The identities required obtain the statement of the linear momentum
equation in \eqref{altLinearMomentum} and to pose the balances of
mass for the concentrations in terms of the chemcial potentials
are used in an essential fashion for the development of stable
numerical schemes with non--negative concentrations \cite{Wa11}.

\section {LC-Enabled Electroosmosis}
Next, we use the model developed in the previous section to study electroosmosis in a nematic liquid crystalline film constrained between two parallel plates.
In \cite{Lavrent2014}, the authors present an approach to generate electrokinetic effects, by using as an electrolyte a liquid crystal with surface-imposed
distortions of molecular orientations. In the presence of the uniform electric field, these variations produce space charge separation 
that triggers electroosmotic flow in the liquid crystal. In particular, for the setup depicted in Fig. \ref{fig:exp}, the director orientation is periodically varying 
in the vertical direction and an AC field is applied in the horizontal direction. It was observed that spatially periodic horizontal flow proceeds along 
the "guiding rails" induced by molecular orientation with the direction of the flow independent of the sign of the field. Indeed, for weak fields, it is known \cite{LavrentovichLazoP2010,LazoLavrent2013,Lazo-Lavrent2014,Lavrent2014} that the driving force of the flow is proportional to the square of the field. The physical reason is simple \cite{Lazo-Lavrent2014,Lavrent2014}. The spatial charge is created at the director distortions by the applied field $E$ and therefore grows linearly with $E$; the Coulomb force, being the product of the charge and the driving field, should thus grow as $E^2$. We are interested in establishing
a simple model of this process.
\begin{figure}[h]
\begin{center}
\includegraphics[height=2in]{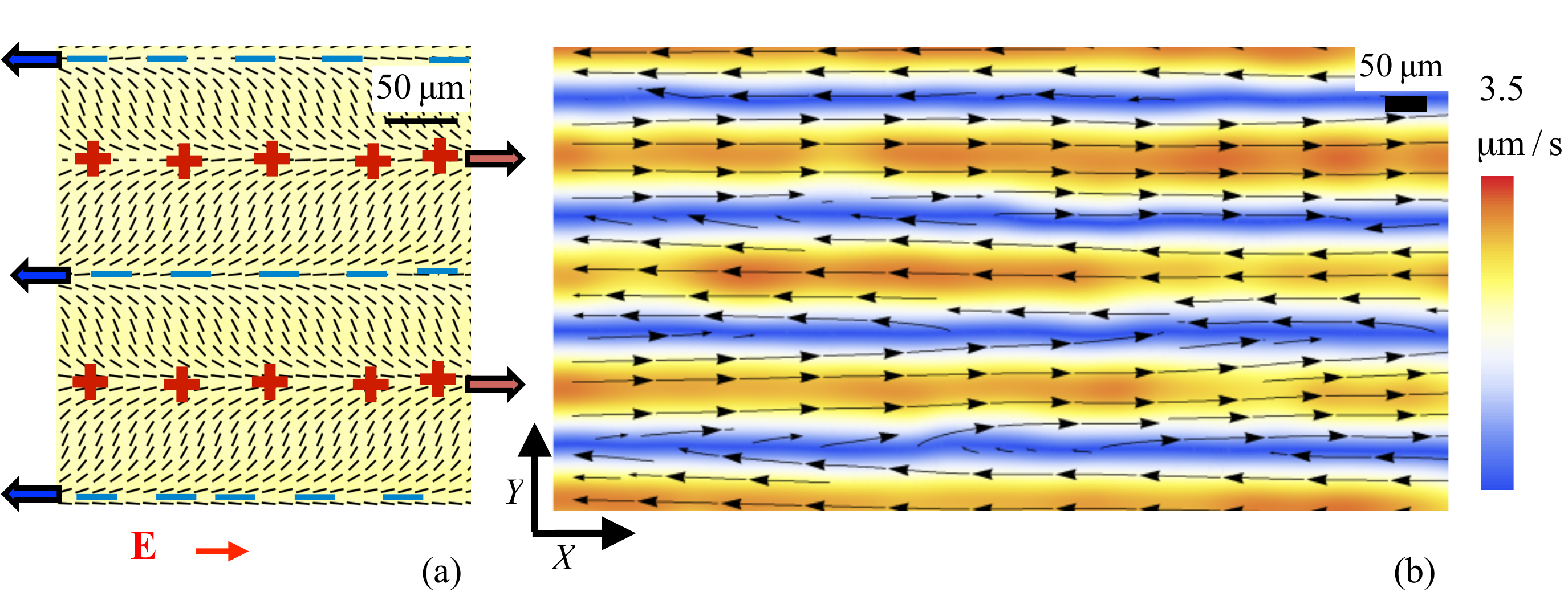}
\caption{Liquid crystal-enabled electroosmotic flows in a flat nematic cell with patterned one-dimensionally periodic director field  \cite{Lavrent2014}. (a) Experimentally imposed director pattern (short black dashes) and a schematics of space charge separation due to director distortions when the electric field is directed from left to right.  Electric conductivity of the nematic is higher along the director than in a direction perpendicular to it.  Clouds of separated positive and negative ions are marked by the "+" and œ"-" symbols, respectively.  Thick arrows show the direction of the electrostatic forces acting on the charge clouds and the direction of local electroosmotic flows. Note that reversal of field polarity reverses the polarity of charges but preserves the directions of the driving electrostatic force and the induced nematic flows.  (b) Experimentally determined map of electroosmotic velocities corresponding to the director pattern in (a).}
\label{fig:exp}
\end{center}
\end{figure}
\subsection{A Simplified Model} In what follows, we will represent the film by a domain $\Omega\times[-h,h]\subset\mathbb{R}^3$, where $\Omega\subset\mathbb{R}^2$ and $0<h\ll1$. The plates are patterned in a way that enforces strong anchoring of the director so that the director orientation varies in the prescribed way along the plates while it remains parallel to the plates. This mode of "pre-patterned" surface anchoring is achieved experimentally through the recently developed plasmonic mask approach \cite{ADMA:ADMA201670081}. Briefly, the photosensitive layers at the inner surfaces of the bounding plates are illuminated by light that passes through the array of narrow elongated nanoslits.  The transmitted light acquires local state of polarization determined by the orientation of nanoslits. The pattern of light polarization is imprinted onto the photosensitive layer; the latter then imposes the alignment pattern onto the director of the adjacent nematic liquid crystal. Since both plates in the assembled cell are irradiated simultaneously, the patterns at the top and bottom plates are identical to each other and impose fixed boundary conditions on both nematic/plate boundaries $\Omega\times\{-h,h\}$. Choosing a coordinate system with the $z$-axis perpendicular to the plates, we impose the Dirichlet condition $\bn|_{\partial\Omega\times\{z\}}=\bn|_{\partial\Omega\times\{h\}}$ for any $z\in[-h,h]$ on the lateral boundary of the film $\partial\Omega\times[-h,h]$.

The primary goal of this section is to demonstrate that the predictions of our model are in qualitative agreement with the experiment. Even though the ideas below apply to a wide range of parameter regimes, here we will only consider a simple setup that is sufficiently anisotropic to replicate experimentally observed behavior. To this end, if we let $\alpha_1=\alpha_2=\alpha_3=\alpha_5=\alpha_6=0$ and $\alpha_4>0$, then by \eqref{eq:gammas} the parameters $\gamma_1=\gamma_2=0$ and, hence by \eqref{viscous-molecular-forceL} the viscous molecular force $\gv\equiv{\bf 0}$. The approximation of isotropic viscosity is justified, as the mechanism of liquid crystal-enabled electrokinetics is not the anisotropy of viscosity but the anisotropy of conductivity (or dielectric permittivity) \cite{Lavrent2014,Lazo-Lavrent2014}.  Viscosity anisotropy renormalizes the velocities of electrokinetic flows, but does not create these flows \cite{Lavrent2014,Lazo-Lavrent2014}. By assuming that the dielectric permittivity anisotropy of the liquid crystal is small and setting $\varepsilon_a=0$, we eliminate direct interaction between the director and the electric field. Although the dielectric anisotropy is expected to play a role similar to anisotropy of conductivity in triggering the liquid crystal-enabled electrokinetics, such a simplification allows us to reflect closely the experimental situations described in \cite{Lavrent2014,Lazo-Lavrent2014}, in which the liquid crystal was formulated to be of zero dielectric anisotropy.

To simplify the model further, we adopt the equal elastic constants approximation $K_1=K_2=K_3=K$, so that 
\[\wof\left(\bn,\nabla\bn\right)=\frac{K}{2}{|\nabla\bn|}^2.\] 
Here we have also eliminated the term in \eqref{Oseen Frank} that corresponds to the elastic constant $K_4$ since this term is a null Lagrangian under the Dirichlet boundary data on the director. Note that the simplifying assumptions made in this paragraph generally do not hold for the experimental system considered in \cite{Lavrent2014}. A more detailed study of this system using the methods described in this section will appear elsewhere.

Given the assumptions above, \eqref{angular momentum} reduces to the harmonic map equation
\begin{equation}
\label{eq:hm1}
\Delta\bn=\gamma\bn
\end{equation}
in $\Omega\times[-h,h]$, where $\bn$ satisfies the Dirichlet conditions on the boundary of the film. Since the director is assumed to be parallel to the plates on the film boundary, we seek a solution of \eqref{eq:hm1} of the form \[\bn(\psi(x,y))=(\cos{\psi(x,y)},\sin{\psi(x,y)},0).\] By substituting this ansatz into \eqref{eq:hm1} we find that 
\begin{equation}
\label{eq:hm2}
\Delta\psi=0
\end{equation}
in $\Omega$. In the remainder of this section we will require the director pattern on the plates to satisfy the equation \eqref{eq:hm2} so that $\bn(x,y)=(\cos{\psi(x,y)},\sin{\psi(x,y)},0)$ is a solution of \eqref{eq:hm1}. In fact, the third component of $\bn$ can be neglected and $\bn$ can be written as \[\bn(\psi(x,y))=(\cos{\psi(x,y)},\sin{\psi(x,y)}).\]

It is now reasonable to look for a solution of the system of governing equations \eqref{anisotropic-NP1}-\eqref{angular momentum} that is independent of the $z$-variable and such that the third component of velocity is identically zero, i.e., $\bv(x,y,t)=(u(x,y,t),v(x,y,t),0)$ or \[\bv(x,y,t)=(u(x,y,t),v(x,y,t)),\] if we drop the trivial component. As an additional simplifying assumption, we consider a case of two ionic species given by the fields \[c_p(x,y,t)\mbox{ and }c_m(x,y,t)\] with $z_p=1$ and $z_m=-1$, respectively. We further select the anisotropic diffusion matrix for both species to be in the form
\begin{equation}
\label{eq:dmtrx}
\mathcal{D}(\psi)=\bar{\mathcal{D}}\left(I+(\lambda-1)\bn(\psi)\otimes\bn(\psi)\right),
\end{equation}
where $\bar{\mathcal{D}}>0$ and the parameter $\lambda\geq0$ determines the strength of anisotropy.

The system \eqref{anisotropic-NP1}-\eqref{angular momentum} now takes the form
\begin{equation}
\label{eq:eqs}
\left\{
\begin{array}{l}
\partderiv{c_p}{t}+ \div(\bv c_p)=\div\left(\mathcal D(\psi)\left(\nabla c_p+\frac{qc_p}{k_B\theta}\nabla\Phi\right)\right),    \\
\partderiv{c_m}{t}+ \div(\bv c_m)=\div\left(\mathcal D(\psi)\left(\nabla c_m-\frac{qc_m}{k_B\theta}\nabla\Phi\right)\right),    \\
-\Delta\Phi=\frac{q}{\varepsilon\varepsilon_0}\left(c_p-c_m\right),   \\
  \rho\left(\frac{\partial{\bv}}{\partial{t}}+\left(\bv\cdot\nabla\right)\bv\right)= -\nabla{p}+\mu\Delta\bv-q\left(c_p-c_m\right)\nabla\Phi, \\
  \div\bv=0, \\
 \Delta\psi=0,
\end{array}
\right.
\end{equation}
in $\Omega$, where $\mu:=\alpha_4/2$. Here we use the symbol $p$ to denote the pressure from \eqref{linear momentum} incremented by the factor $k_B\theta\left(c_p+c_m\right)+\frac{K}{2}{|\nabla\psi|}^2$. Specializing further to a rectangular domain $\Omega=[-L,L]\times[-W,W]$, we impose the boundary conditions
\begin{equation}
\label{eq:bcs}
\left\{
\begin{array}{ll}
\bv={\mathbf 0}  & \mathrm{on}\ \{-L,L\}\times[-W,W],    \\
v=0\mbox{ and }\frac{\partial u}{\partial y}=0  & \mathrm{on}\ [-L,L]\times\{-W,W\},    \\
\left(\nabla c_p+\frac{qc_p}{k_B\theta}\nabla\Phi\right)\cdot\mathcal D(\psi)\boldsymbol\nu=0 & \mathrm{on}\ \partial\Omega,     \\
\left(\nabla c_m-\frac{qc_m}{k_B\theta}\nabla\Phi\right)\cdot\mathcal D(\psi)\boldsymbol\nu=0  & \mathrm{on}\ \partial\Omega,   \\
\frac{\partial\Phi}{\partial y}(x,\pm W,t)=0, &  \\
\Phi(\pm L,y,t)=\pm\Phi_0(t).
\end{array}
\right.
\end{equation}
Note that, in writing the impenetrability conditions on $c_p$ and $c_m$ in \eqref{eq:bcs}, we took advantage of the symmetry of the diffusion matrix.
Furthermore, the second equation in \eqref{eq:bcs} corresponds to a perfect slip condition on the lateral components of the boundary $\partial\Omega$. Here the conditions on $[-L,L]\times\{-W,W\}$ are equivalent to imposing periodic boundary conditions on the solution of \eqref{eq:eqs} corresponding to periodic director distributions discussed below.
\begin{remark}
Since $\bn(x,y)=(\cos{\psi(x,y)},\sin{\psi(x,y)})$, then
\begin{equation}
\label{eq:stress_n}
\nabla\bn^T\nabla\bn=\nabla\psi\otimes\nabla\psi=\frac{{|\nabla\psi|}^2}{2}I+\boldsymbol\tau,
\end{equation}
where, by \eqref{eq:hm2}, the deviatoric stress tensor $\boldsymbol\tau$ is divergence-free:
$\div{\boldsymbol\tau}=0.$
\end{remark}

\subsection{Nondimensionalization}
Next, we nondimensionalize the system \eqref{eq:eqs}-\eqref{eq:bcs} as follows. Let
\begin{equation}
\label{eq:nondims}
\tilde \bx=\frac{\bx}{\bar{W}},\quad \tilde\bv=\frac{\bv}{\bar u},\quad \tilde t=\frac{\bar ut}{\bar{W}},\quad \tilde c_p=\frac{c_p}{\bar c},\ \tilde c_m=\frac{c_m}{\bar c},\quad \tilde\Phi=\frac{\Phi}{\bar\Phi},\quad \tilde p=\frac{p}{\bar p},
\end{equation}
where $\bar f$ denotes the characteristic value of a given quantity $f$. We let $\bar\Phi=\bar{W}L^{-1}\|\Phi_0\|_\infty=\bar{W}\|E_0\|_\infty$, where $E_0$ represents the strength of the electric field between the electrodes. Following
\cite{Lavrent2014}, assume that
\begin{equation}
\label{eq:scu}
\bar{u}=\frac{\varepsilon\varepsilon_0\bar{\Phi}^2}{\mu \bar{W}}.
\end{equation}
By denoting $\tilde{\mathcal{D}}(\psi)={\bar{\mathcal{D}}}^{-1}\left(\mathcal{D}_{ij}(\psi)\right)$, dropping all tildes for notational convenience, and setting $f_{,r}:=\partial{f}/\partial{r}$ for any $f$ and $r$, we obtain the system of nondimensional equations 
\begin{equation}
\label{eq:eqsnd}
\left\{
\begin{aligned}
&  \Pe\left(c_{p,t}+uc_{p,x}+vc_{p,y}\right)=\left(\mathcal D_{11}(\psi)\left(c_{p,x}+\FN c_p\Phi_{,x}\right)\right)_{,x} \\ & \quad\quad\quad\quad+\left(\mathcal D_{12}(\psi)\left(c_{p,x}+\FN c_p\Phi_{,x}\right)\right)_{,y}+\left(\mathcal D_{12}(\psi)\left(c_{p,y}+\FN c_p\Phi_{,y}\right)\right)_{,x}    \\
& \quad\quad\quad\quad\quad\quad\quad\quad+\left(\mathcal D_{22}(\psi)\left(c_{p,y}+\FN c_p\Phi_{,y}\right)\right)_{,y}, \\
& \Pe\left(c_{m,t}+uc_{m,x}+vc_{m,y}\right)=\left(\mathcal D_{11}(\psi)\left(c_{m,x}-\FN c_m\Phi_{,x}\right)\right)_{,x} \\ & \quad\quad\quad\quad+\left(\mathcal D_{12}(\psi)\left(c_{m,x}-\FN c_m\Phi_{,x}\right)\right)_{,y}+\left(\mathcal D_{12}(\psi)\left(c_{m,y}-\FN c_m\Phi_{,y}\right)\right)_{,x}    \\
& \quad\quad\quad\quad\quad\quad\quad\quad+\left(\mathcal D_{22}(\psi)\left(c_{m,y}-\FN c_m\Phi_{,y}\right)\right)_{,y}, \\
 -&\Phi_{,xx}-\Phi_{,yy}=\BN \left(c_p-c_m\right),   \\
& \psi_{,xx}+\psi_{,yy}=0, \\
&   \Rey\left(u_{,t}+uu_{,x}+vu_{,y}\right)= -p_{,x}+u_{,xx}+u_{,yy}-\BN\left(c_p-c_m\right)\Phi_{,x}, \\
&  \Rey\left(v_{,t}+uv_{,x}+vv_{,y}\right)= -p_{,y}+v_{,xx}+v_{,yy}-\BN\left(c_p-c_m\right)\Phi_{,y}, \\
&  u_x+v_y=0,
\end{aligned}
\right.
\end{equation}
in $\Omega=\left[-\varepsilon^{-1},\varepsilon^{-1}\right]\times[-a,a]$. Here we set
\begin{equation}
\label{eq:nondimgrs}
\Rey=\frac{\rho\bar u\bar{W}}{\mu},\quad \Pe=\frac{\bar u \bar{W}}{\bar{\mathcal{D}}},\quad \varepsilon=\frac{\bar{W}}{L},\quad a=\frac{W}{\bar{W}},\quad \FN=\frac{q\bar{\Phi}}{k_B\theta},\quad \BN=\frac{q\bar c \bar{W}^2}{\varepsilon\varepsilon_0\bar\Phi}
\end{equation}
and assume that $\bar p=\mu\bar{u}W^{-1}$. The boundary conditions can now be written as
\begin{equation}
\label{eq:bcsnd}
\left\{
\begin{array}{ll}
u=v=0 & \mathrm{on}\ \left\{-\varepsilon^{-1},\varepsilon^{-1}\right\}\times[-a,a],    \\
v=0\mbox{ and }u_{,y}=0  &  \mathrm{on}\ \left[-\varepsilon^{-1},\varepsilon^{-1}\right]\times\{-a,a\},    \\
\mathcal D_{11}(\psi)\left(c_{p,x}+\FN c_p\Phi_{,x}\right)+\mathcal D_{12}(\psi)\left(c_{p,y}+\FN c_p\Phi_{,y}\right)=0 & \mathrm{on}\ \left\{-\varepsilon^{-1},\varepsilon^{-1}\right\}\times[-a,a],     \\ \smallskip
\mathcal D_{11}(\psi)\left(c_{p,x}+\FN c_p\Phi_{,x}\right)+\mathcal D_{12}(\psi)\left(c_{p,y}+\FN c_p\Phi_{,y}\right)=0  & \mathrm{on}\ \left\{-\varepsilon^{-1},\varepsilon^{-1}\right\}\times[-a,a],   \\ \smallskip
\mathcal D_{12}(\psi)\left(c_{px}+\FN c_p\Phi_x\right)+\mathcal D_{22}(\psi)\left(c_{py}+\FN c_p\Phi_y\right)=0  & \mathrm{on}\ \left[-\varepsilon^{-1},\varepsilon^{-1}\right]\times\left\{-a,a\right\},   \\ \smallskip
{\mathcal{D}}_{12}(\psi)\left(c_{mx}-\FN c_m\Phi_x\right)+{\mathcal{D}}_{22}(\psi)\left(c_{my}-\FN c_m\Phi_y\right)=0  & \mathrm{on}\ \left[-\varepsilon^{-1},\varepsilon^{-1}\right]\times\left\{-a,a\right\},   \\ \smallskip
\Phi_y(x,\pm a,t)=0, &  \\
\Phi\left(\pm\varepsilon^{-1},y,t\right)=\pm\Phi_0(t), 
\end{array}
\right.
\end{equation}
where we set $\tilde\Phi_0(t)=\frac{\Phi_0(t)}{\bar\Phi}$ and drop the tilde. From now on---unless specified otherwise---we will work with the nondimensional problem \eqref{eq:eqsnd}-\eqref{eq:bcsnd}.

\subsection{Periodic Flow Pattern}
Suppose now that the director field follows a periodic stripe pattern with the stripes being parallel to the $x$-axis. In nondimensional coordinates this can be modeled, for example, by assuming that $a=\pi n$ for some $n\in\mathbb N$ and setting $\psi=\frac{y}{2}$. One can immediately observe that this function satisfies the fourth equation in \eqref{eq:eqsnd}. By a direct computation we also have
\begin{equation}
\label{eq:diifc}
\mathcal{D}=\frac{1}{2}\left((\lambda+1)I+(\lambda-1)
\left(
\begin{array}{cc}
 \cos{y} &   \sin{y}    \\
  \sin{y} &   -\cos{y}    
\end{array}
\right)
\right).
\end{equation}

Our setup can be associated with an electrochemical experiment, in which the right---$\left\{\varepsilon^{-1}\right\}\times[-\pi n,\pi n]$---and the left---$\left\{-\varepsilon^{-1}\right\}\times[-\pi n,\pi n]$---components of the boundary are identified as a positive and a negative electrode, respectively. Here the rest of the boundary is assumed to be electrically insulated. As the positive ions will be attracted to the negative electrode and vice versa, the boundary layers would form next to the electrodes that would subsequently suppress both the potential difference and flow in the nematic electrolyte as long as the electrodes potentials remain fixed. In an experiment, this is circumvented by applying the AC instead of the DC field making the corresponding problem inherently transient. Here we will assume that the parameter $\varepsilon$ is small and that the flux of ions on the timescale of the flow is not large enough to significantly affect the boundary layers. We will thus solve the "outer problem" away from the electrodes and set $\Phi_0$ to be constant in time and equal to the half of the potential difference between the "matching regions" corresponding to the edges of the right and left boundary layers. Since we are not solving the equations inside the boundary layer, we can then also ignore the second and the third boundary conditions in \eqref{eq:bcsnd}.

We now seek a solution of \eqref{eq:eqsnd}-\eqref{eq:bcsnd} in the form
\[c_p=c_p(y),\quad c_m=c_m(y),\ \bv=(u(y),0),\ \Phi=x+\phi(y).\]
Then the incompressibility condition trivially holds and the system \eqref{eq:eqsnd}-\eqref{eq:bcsnd} reduces to 
\begin{equation}
\label{eq:eqs1d}
\left\{
\begin{aligned}
&  \left(\FN\,\mathcal D_{12}(\psi)c_p+\mathcal D_{22}(\psi)\left(c_{p,y}+\FN c_p\phi_{,y}\right)\right)_{,y}=0, \\ 
 &\left(-\FN\,\mathcal D_{12}(\psi)c_m+\mathcal D_{22}(\psi)\left(c_{m,y}-\FN c_m\phi_{,y}\right)\right)_{,y}=0, \\
&\phi_{,yy}=-\BN \left(c_p-c_m\right),   \\
&   u_{,yy}=\BN\left(c_p-c_m\right), \\
&p_{,y}=-\BN\left(c_p-c_m\right)\phi_y, \\
&  p_{,x}=0, \\
& \int_{-\pi n}^{\pi n} c_p\,dy=\int_{-\pi n}^{\pi n} c_m\,dy=2\pi n, \\
& \int_{-\pi n}^{\pi n} u\,dy=0, 
\end{aligned}
\right.
\end{equation}
subject to the boundary conditions
\begin{equation}
\label{eq:bcs1d}
\left\{
\begin{array}{l}
u_{,y}(\pm \pi n)=0,    \\ 
\FN\,\mathcal D_{12}(\psi(\pm \pi n))c_p(\pm \pi n)+\mathcal D_{22}(\psi(\pm \pi n))c_{p,y}(\pm \pi n)=0,     \\ 
\FN\,\mathcal D_{12}(\psi(\pm \pi n))c_m(\pm \pi n)-\mathcal D_{22}(\psi(\pm \pi n))c_{m,y}(\pm \pi n)=0,    \\ 
\phi_{,y}(\pm \pi n)=0,
\end{array}
\right.
\end{equation}
where $f_{,y}=df/dy$ for any function $f$ of $y$. Note that we have added three integral conditions to the problem in order to take into account conservation of mass of both species as well as to ensure that there is a zero net flow across each cross-section of the domain. Indeed, since the time-dependent problem has now been replaced by the stationary problem, we have to impose the condition that the total masses of both ionic species are the same as what they were at the initial time for the time-dependent problem. 

The fourth and fifth equations in \eqref{eq:eqs1d} are used to determine pressure. By combining the third and the fourth equations in \eqref{eq:eqs1d}, we find that 
\begin{equation}
\label{eq:n1}
u_{,yy}=-\phi_{,yy},
\end{equation}
so that the nematic liquid crystal velocity is found once the electric field has been computed. To this end, observe that by integrating the first two equations in \eqref{eq:eqs1d} and using \eqref{eq:bcs1d}, we have
\[\FN\,\mathcal D_{12}(\psi)c_p+\mathcal D_{22}(\psi)\left(c_{p,y}+\FN c_p\phi_{,y}\right)=0,\quad\quad\quad-\FN\,\mathcal D_{12}(\psi)c_m+\mathcal D_{22}(\psi)\left(c_{m,y}-\FN c_m\phi_{,y}\right)=0\]
on $[-\pi n,\pi n]$. Dividing these equations by $\mathcal{D}_{22}(\psi)c_p$ and $\mathcal{D}_{22}(\psi)c_m$, respectively, and using \eqref{eq:diifc} gives
\begin{equation}
\label{eq:r31}
\left(\log{c_p}+\FN\phi\right)_{,y}=-\frac{\FN\,\mathcal D_{12}(\psi)}{\mathcal D_{22}(\psi)}=-\FN\beta_{,y},\quad\left(\log{c_m}-\FN\phi\right)_{,y}=\frac{\FN\,\mathcal D_{12}(\psi)}{\mathcal D_{22}(\psi)}=\FN\beta_{,y}
\end{equation}
on $[-\pi n,\pi n]$ where
\begin{equation}
\label{eq:r34}
\beta(y)=\log{\left(\lambda+1-(\lambda-1)\cos{y}\right)}.
\end{equation}
The system of governing equations can now be reduced to a single equation for $r=\log{c_p}$. By taking the derivative of the first equation in \eqref{eq:r31}, we obtain 
\begin{equation}
\label{eq:r35}
r_{,yy}+\FN\phi_{,yy}=-{\FN}\beta_{,yy}(y).
\end{equation}
Since from \eqref{eq:r33} it follows that $c_m=c_0^2/c_p=c_0^2e^{-r}$, using the third equation in \eqref{eq:eqs1d} gives 
\begin{equation}
\label{eq:r36}
r_{,yy}-\FN\,\BN \left(e^r-c_0^2e^{-r}\right)=-{\FN}\beta_{,yy}(y).
\end{equation}
This equation should be supplemented by the integral conditions from \eqref{eq:eqs1d} that in terms of $r$ take the form
\begin{equation}
\label{eq:r37}
\int_{-\pi n}^{\pi n}e^r\,dy=c_0^2\int_{-\pi n}^{\pi n}e^{-r}\,dy=2\pi n.
\end{equation}
Even though they are significantly simpler than the original system of partial differential equations, both the problem \eqref{eq:phi} and the problem \eqref{eq:r36}-\eqref{eq:r37} still need to be solved either numerically or by using asymptotic expansions, provided that a suitable small parameter can be identified.

\begin{remark}
Alternatively, the solution procedure for the system \eqref{eq:eqs1d}-\eqref{eq:bcs1d} can essentially be reduced to solving a single second order nonlinear ODE for the potential.  Indeed, integrating \eqref{eq:r31}, we find
\[c_p=c_p^0e^{-F(\phi+\beta)}\mbox{   and   }c_m=c_m^0e^{F(\phi+\beta)},\]
where $c_p^0$ and $c_m^0$ are arbitrary positive constants. Since $\phi$ is determined up to an arbitrary constant, we can replace $\phi$ by $\phi-\phi_0$ and choose $\phi_0$ so that
\[c_p^0e^{\FN\phi_0}=c_m^0e^{-\FN\phi_0}=c_0.\]
It follows that 
\begin{equation}
\label{eq:r33}
c_p=c_0e^{-F(\phi+\beta)}\mbox{   and   }c_m=c_0e^{F(\phi+\beta)}.
\end{equation}
Substituting these expressions into the third equation in \eqref{eq:eqs1d}, gives the problem
\begin{equation}
\label{eq:phi}
\left\{
\begin{array}{ll}
\phi_{,yy}=c_0\BN\left(e^{\FN(\phi+\beta)}-e^{-\FN(\phi+\beta)}\right), & y\in(-\pi n,\pi n), \\
\phi_{,y}(\pm \pi n)=0,
\end{array}
\right.
\end{equation}
satisfied by $\phi$.

\end{remark} 

\subsubsection{Asymptotic Solutions}
The behavior of solutions of \eqref{eq:r36}-\eqref{eq:r37} is determined by the sizes of nondimensional groups $\BN$ and $\FN$. In the experimental setup considered in 
\cite{LazoLavrent2013}, the physical parameters had the following values
\begin{equation}
\label{eq:dun}
\begin{array}{llll}
 L=1\cdot10^{-2}\,\mathrm{m}, & W=5\cdot10^{-4}\,\mathrm{m},  & H=5\cdot10^{-5}\,\mathrm{m}, & \lambda=1.417, \\ \smallskip \bar{\mathcal{D}}=4.89\cdot10^{-11}\,\mathrm{m}^2/\mathrm{s}, & \bar{c}=1\cdot10^{19}\,\mathrm{m}^{-3},  & q=1.6\cdot10^{-19}\,\mathrm{C},  & \mu=0.832\,\mathrm{Pa}\cdot \mathrm{s}, \\ \smallskip \varepsilon\varepsilon_0=5.32\cdot10^{-11}\,\mathrm{F}/\mathrm{m}, & \rho=1\cdot10^{3}\, \mathrm{kg}/\mathrm{m}^3, & \Phi_0=400\,\mathrm{V}, & \bar{W}=5\cdot10^{-5}\,\mathrm{m}.   
\end{array}
\end{equation}
These yield
\begin{equation}
\label{eq:bars}
\bar{u}=5\cdot10^{-6}\, \mathrm{m}/\mathrm{s}\mbox{ and }\bar{\Phi}=2\,\mathrm{V},
\end{equation}
and
\begin{equation}
\label{eq:ndg}
\begin{array}{lllll}
\Rey=3\cdot10^{-7},& \Pe=5,& \varepsilon=5\cdot10^{-3},& \FN=79, & \BN=37.   
\end{array}
\end{equation}

According to these observations, here we will assume that the nondimensional groups $\FN,\BN\gg1$, so that we can take advantage of a natural small parameter $\frac{1}{\FN}$ in order to solve the problem \eqref{eq:r36}-\eqref{eq:r37} asymptotically. We thus suppose that $\delta:=\frac{1}{\FN}$, where $\delta\ll1$ and use \eqref{eq:ndg} to set $\BN=b/\delta$ where $b=O(1)$. Then \eqref{eq:r36}-\eqref{eq:r37} takes the form
\begin{equation}
\label{eq:r36.1}
\left\{
\begin{array}{ll}
\displaystyle\delta r_{,yy}-\frac{b}{\delta}\left(e^r-c_0^2e^{-r}\right)=-\beta_{,yy},  &   y\in(-\pi n,\pi n),  \\ \\
\displaystyle\int_{-\pi n}^{\pi n}e^r\,dy=c_0^2\int_{-\pi n}^{\pi n}e^{-r}\,dy=2\pi n. &
\end{array}
\right.
\end{equation}
Assuming that 
 \[r=\delta r_1+O\left(\delta^2\right),\quad c_0=1-\frac{c_1}{2}\delta+O\left(\delta^2\right),\]
 we can rewrite \eqref{eq:r36.1} as
\[\left\{
\begin{array}{ll}
-2b\,r_1=-\beta_{,yy}+b\,c_1,  &   y\in(-\pi n,\pi n),\\ \\
\displaystyle\int_{-\pi n}^{\pi n}r_1\,dy=2\pi nc_1+\int_{-\pi n}^{\pi n}r_1\,dy=0,  &    
\end{array}
\right.
\]
to the leading order in $\delta$. Then, it immediately follows that $c_1=0$ and that $r_1$ is given by
\begin{equation}
\label{eq:r39}
r_1(y)=\frac{1}{2b}\beta_{,yy}.
\end{equation} 
From \eqref{eq:r31} and \eqref{eq:r39}, setting $\phi(y)=\phi_0(y)+O\left(\delta\right),$ we obtain
\begin{equation}
\label{eq:r40}
\left\{
\begin{array}{ll}
\phi_{0,y}=-\beta_{,y},  &   y\in(-\pi n,\pi n),\\
\phi_{0,y}(\pm \pi n)=0.  &    
\end{array}
\right.
\end{equation}
We conclude that, up to an arbitrary constant of integration,
\begin{equation}
\label{eq:r41}
\phi_0(y)=-\beta(y).
\end{equation}
Finally, from \eqref{eq:n1}, the expression for the leading term in velocity $u(y)=u_0(y)+O\left(\delta\right)$ takes the form
\[u_0(y)=\beta(y)+C_u^1y+C_u^2,\]
where $C_u^1$ and $C_u^2$ are arbitrary constants. Since by \eqref{eq:bcs1d} we have that $u_{0y}(\pm \pi n)=0,$ then using the integral condition on velocity in \eqref{eq:eqs1d} it follows that
\[C_u^1=0,\quad\quad C_u^2=-\frac{1}{2\pi n}\int_{-\pi n}^{\pi n}\beta\,dy,\]
hence
\begin{equation}
\label{eq:r42}
u_0(y)=\beta(y)-\frac{1}{2\pi n}\int_{-\pi n}^{\pi n}\beta\,dy.
\end{equation}
To summarize the asymptotic results above, up to terms $O\left(\delta^2\right)$, we have 
\begin{align}
&c_p=1+\frac{1}{2\BN}\beta_{,yy}, \label{eq:r43.1} \\
&c_m=1-\frac{1}{2\BN}\beta_{,yy}, \label{eq:r43.2} \\
&\phi=-\beta, \label{eq:r43.3} \\
&u=\beta-\frac{1}{2\pi n}\int_{-\pi n}^{\pi n}\beta\,dy. \label{eq:r43.4} 
\end{align}
Here we have used $\delta=b/\BN$ and the facts that $c_p=e^r=1+\delta r_1+O\left(\delta^2\right)$ and $c_m=c_0^2e^{-r}=1-\delta r_1+O\left(\delta^2\right)$. 

Note that, since the scaling \eqref{eq:scu} for $u$ was chosen to be quadratic in the applied field while the product of $\FN$ and $\BN$ in \eqref{eq:nondimgrs} is independent of the field, the flow velocity $u$ is {\it quadratic} in 
the field when both are expressed in dimensional units. Thus reversing the direction of the field {\it would not} result in flow reversal. Note that 
  any solution obtained by solving \eqref{eq:r36.1}, also solves the original system \eqref{eq:eqs1d}, although it does not satisfy all of the boundary conditions in \eqref{eq:bcs1d}.

\subsubsection{Applicability of Asymptotic Solutions: Comparison with Numerical Results}
In Figs \ref{fig:cp-cm}-\ref{fig:phi}, we have plotted in dimensional units both the asymptotic solutions \eqref{eq:r43.1}-\eqref{eq:r43.4} as well as the numerical solutions of \eqref{eq:eqs1d}-\eqref{eq:bcs1d} for the parameter values given by \eqref{eq:dun}-\eqref{eq:ndg}. 
\begin{figure}[H]
\begin{center}
\includegraphics[height=2.5in]{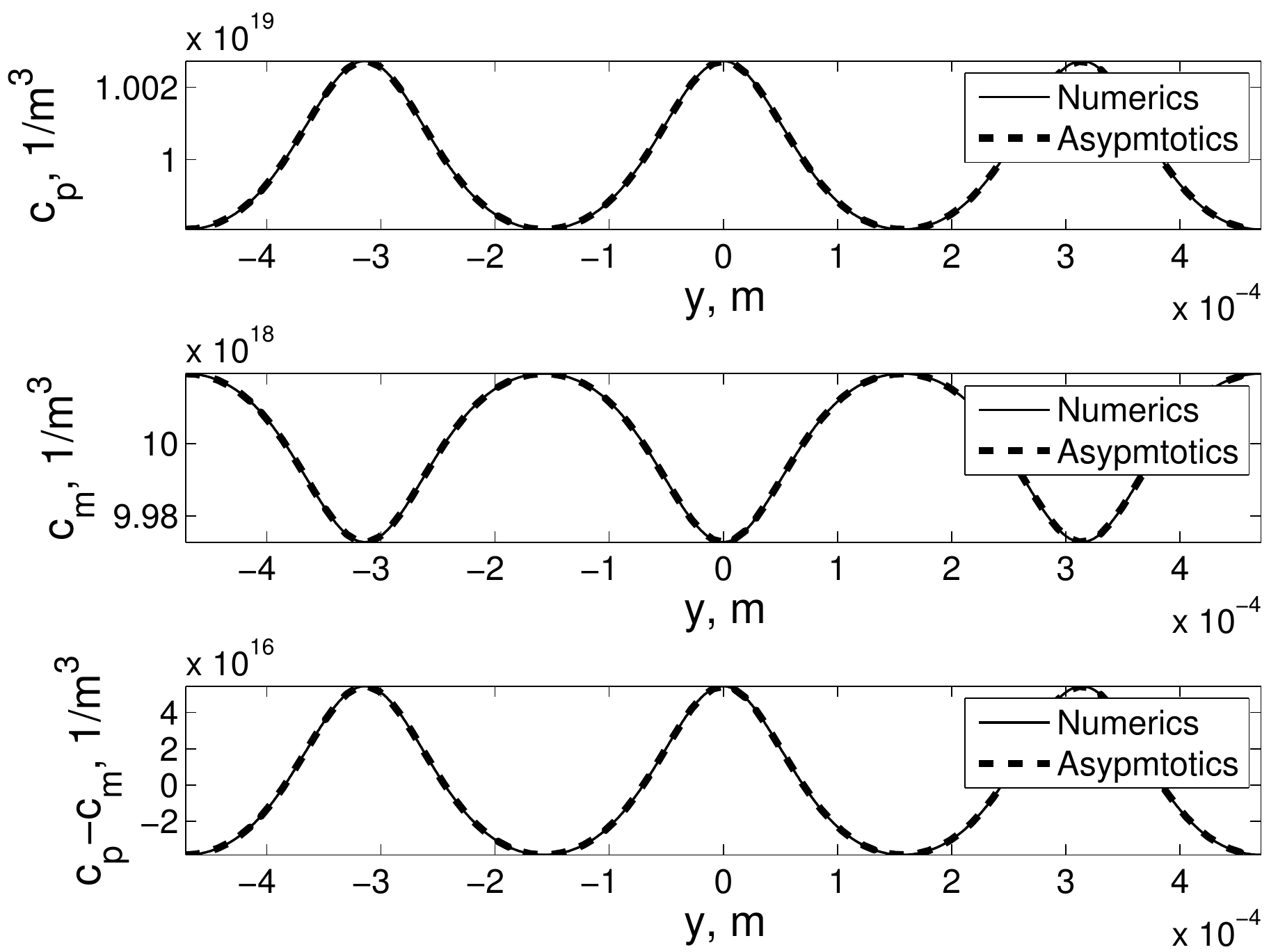}
\caption{Concentration of the positive ions (top), concentration of the negative ions (middle), and space charge distribution (bottom) for the parameters in \eqref{eq:dun}.}
\label{fig:cp-cm}
\end{center}
\end{figure}
The numerical solution was obtained by solving \eqref{eq:r36}-\eqref{eq:r37} in MATLAB \cite{MATLAB} using the standard boundary value problem solver. The results show an excellent match between the corresponding numerical and asymptotic solution fields. In Fig. \ref{fig:phi-s} we have plotted on a logarithmic scale the maximum magnitude of the flow velocity as a function of the applied potential. The dependence on the field is clearly quadratic since  $u_{max}\sim\Phi_0^{2}$. Note that the solutions in Figs \ref{fig:cp-cm}-\ref{fig:phi} are qualitatively similar to the experimental results in Fig. \ref{fig:exp} in that both demonstrate periodicity of velocity and charge distribution patterns. Further, even though some of the assumptions we have made in this section should rule out quantitative similarity between the theory and the experiment, the velocities predicted by the simplified model are of the same order of magnitude as those observed in the experiment (cf. Fig. \ref{fig:exp}).  
\begin{figure}[H]
\begin{center}
\includegraphics[height=2in]{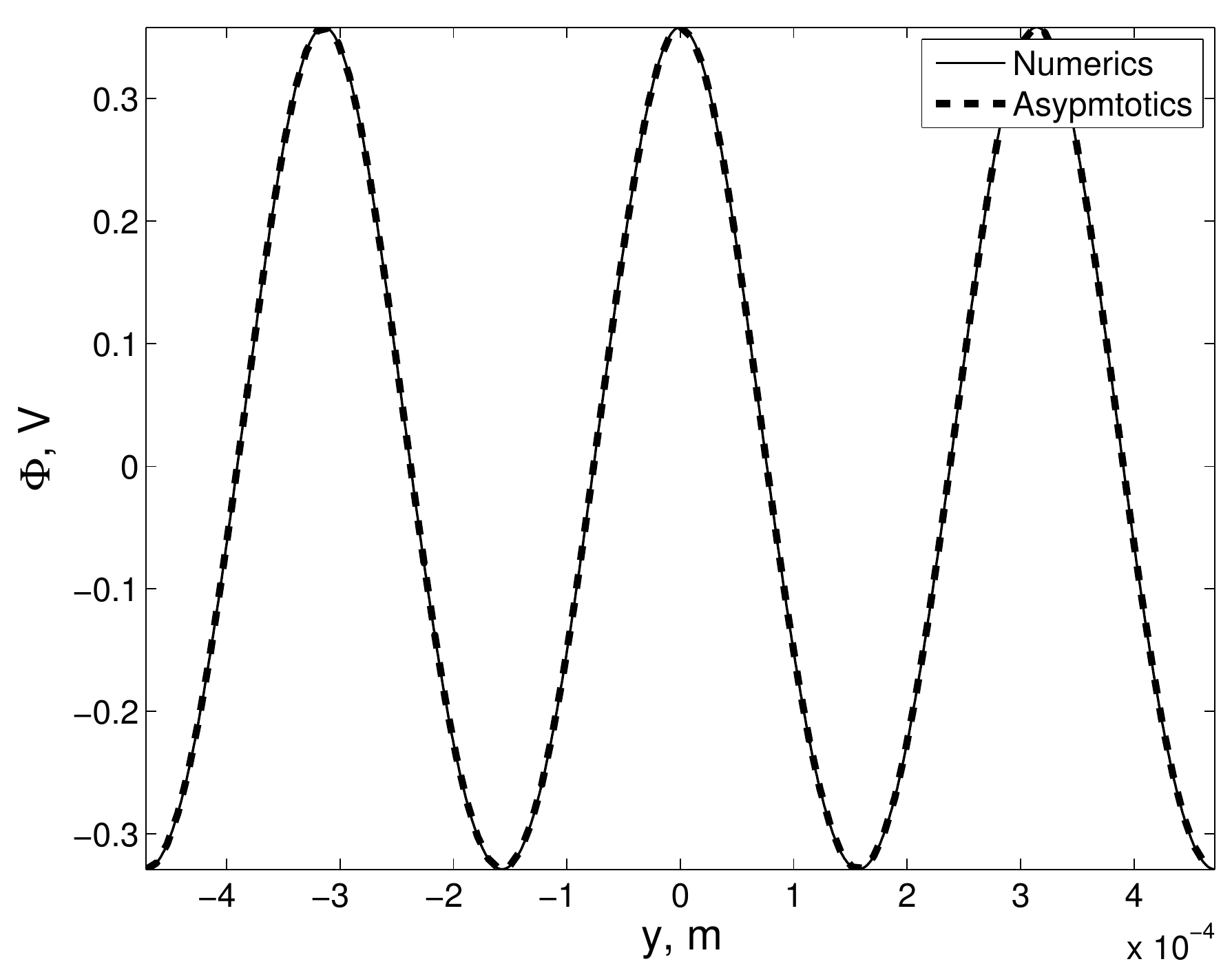}\qquad \includegraphics[height=2.07in]{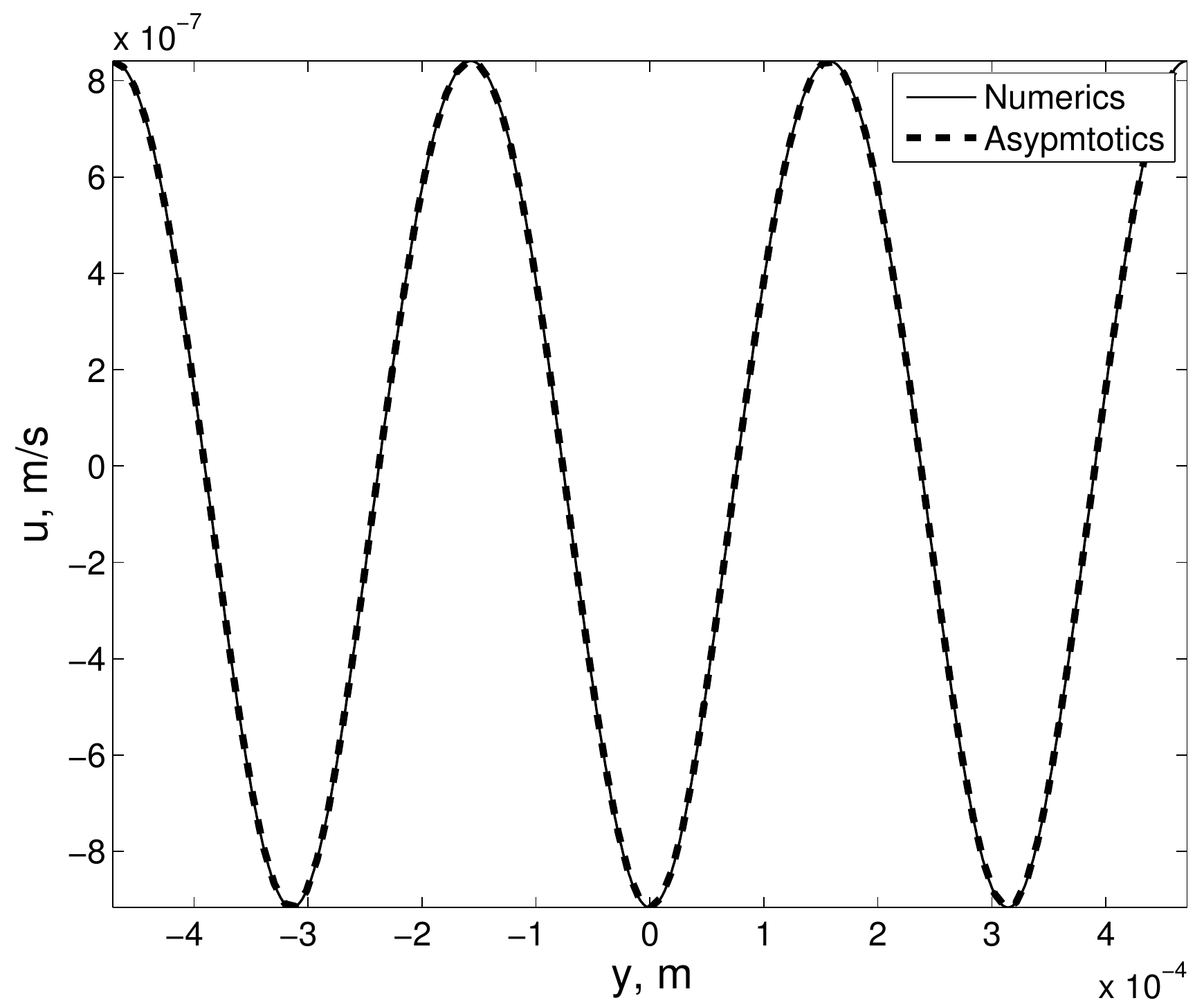}
\caption{Electric potential (left) and flow velocity (right) for the parameters in \eqref{eq:dun}.}
\label{fig:phi}
\end{center}
\end{figure}
\begin{figure}[H]
\begin{center}
\includegraphics[height=2in]{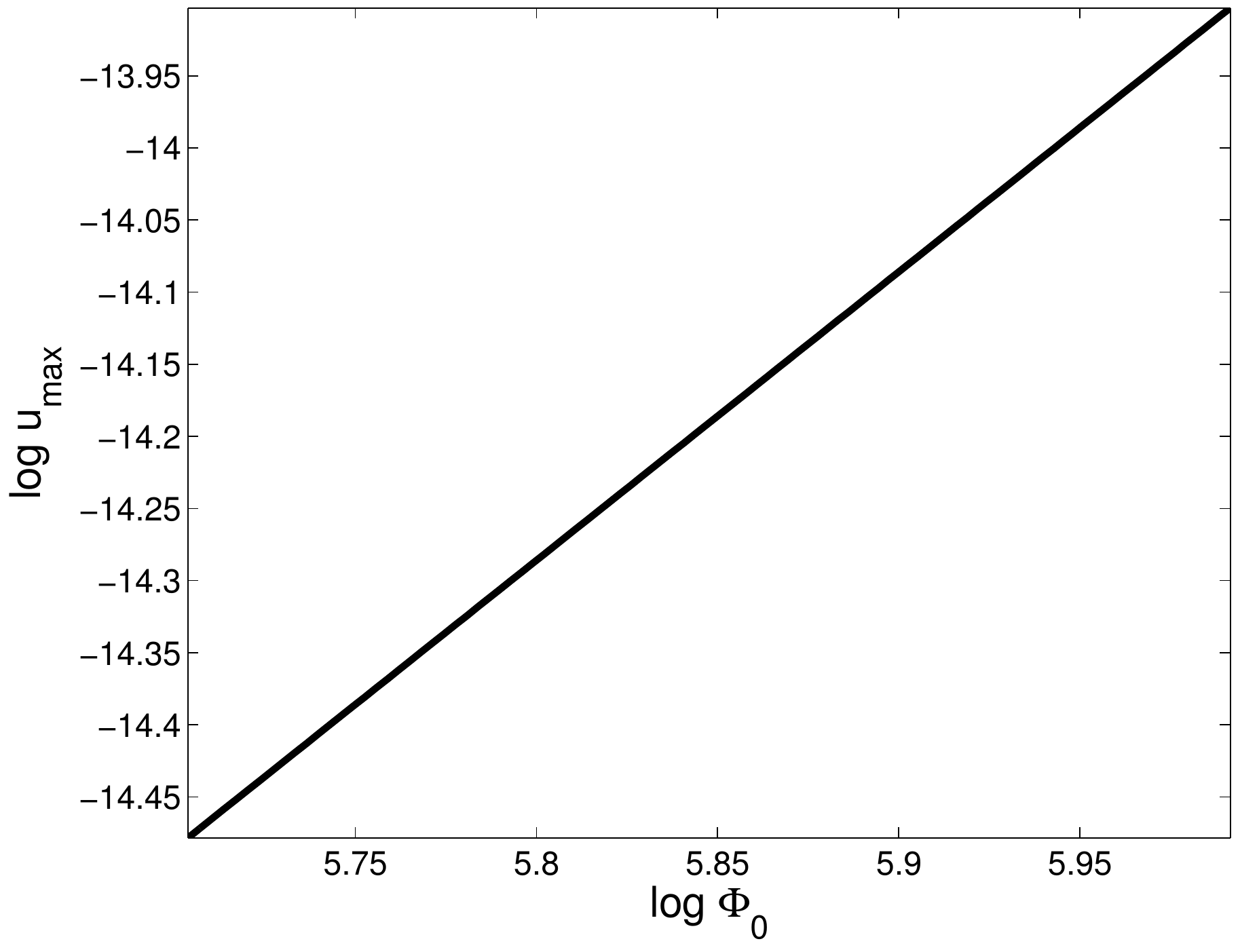}
\caption{Maximum flow velocity as a function of $\Phi_0$ on a logarithmic scale when $\bar{c}=1\cdot10^{19}\,\mathrm{m}^{-3}$ and $\Phi_0\in[300,400]$ V. The slope of the graph is equal to $2$.}
\label{fig:phi-s}
\end{center}
\end{figure}
In general, the behavior of solutions critically depends on the sizes of the nondimensional groups $\FN$ and $\BN$ defined in \eqref{eq:nondimgrs}. The parameter $\FN$ is equal to the ratio of the electrostatic potential energy of an ion and its thermal energy. Correspondingly, when $\FN$ is large, electrostatic force dominates over diffusion and the latter can be ignored. To understand the role of $\BN$, observe that if the applied field is strong enough to cause separation of all charges in the nematic, the field along the $y$-direction can be interpreted as a field in a capacitor 
with an area charge density of $\bar{c}\bar{W}$ and the distance between the capacitor plates equal to $\bar{W}$. 
The potential difference between the plates of such capacitor is equal to $\bar{c}\bar{W}^2/\varepsilon\varepsilon_0$ and we conclude that the nondimensional 
group $\BN$ is equal to the ratio between the characteristic potential $\bar\Phi$ and the maximum potential difference that can be supported by the system via separation
of charges. Then, if $\BN\gg1$---as in the experiment in 
\cite{LazoLavrent2013}---there are enough charges in the system to support quadratic growth of the flow velocity, as the applied field increases. 
If this parameter is small, however, then all of the available charges are expected to move---according to their sign---to different locations prescribed by the distribution of the nematic director. The flow then should become proportional to the magnitude of the field. 

\begin{figure}[H]
\begin{center}
\includegraphics[height=2.5in]{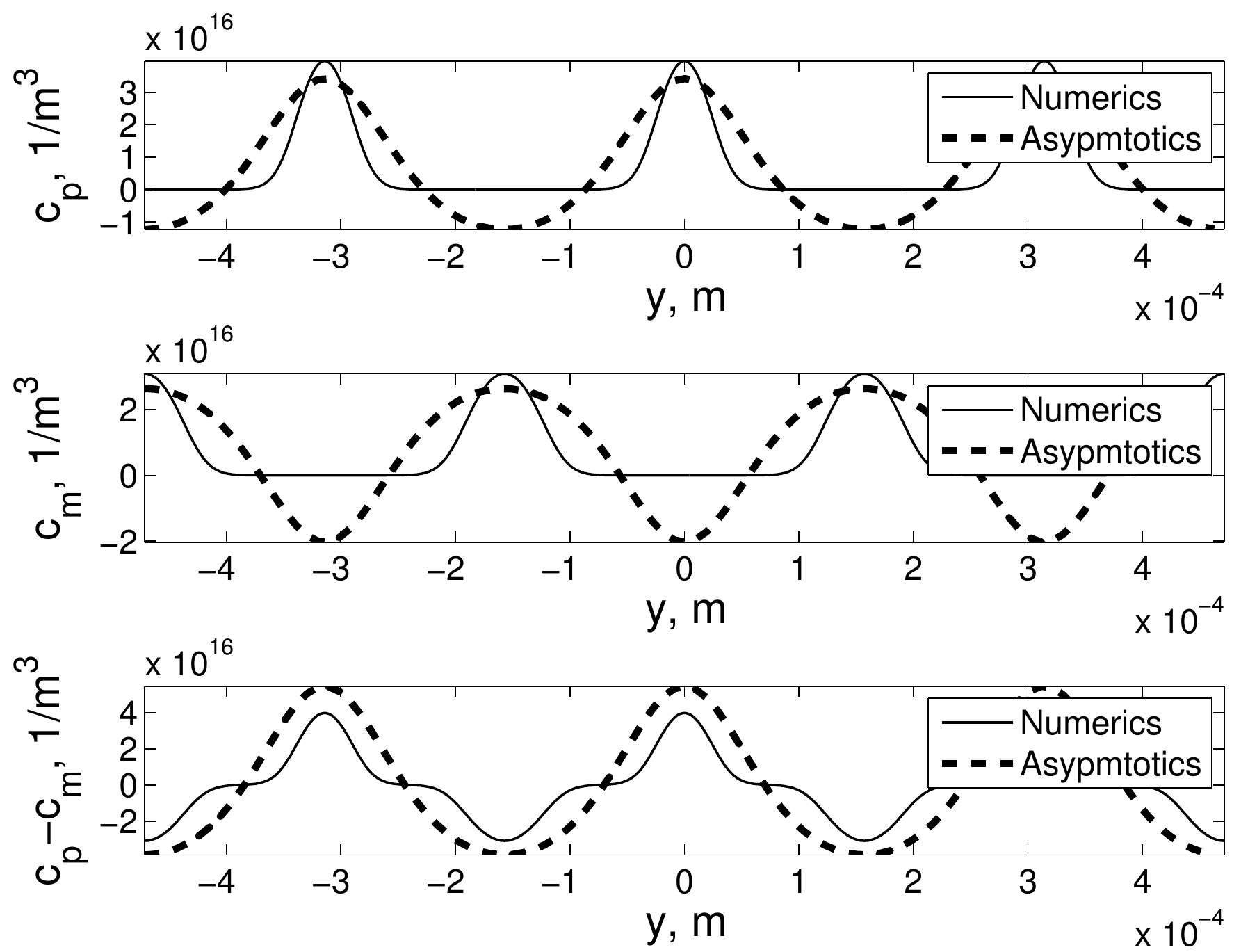}
\caption{Concentration of the positive ions (top), concentration of the negative ions (middle), and space charge distribution (bottom) when $\bar{c}=7\cdot10^{15}\,\mathrm{m}^{-3}$ and the remaining parameters are as in \eqref{eq:dun}.}
\label{fig:cp-cm1}
\end{center}
\end{figure}
\begin{figure}[H]
\begin{center}
\includegraphics[height=2in]{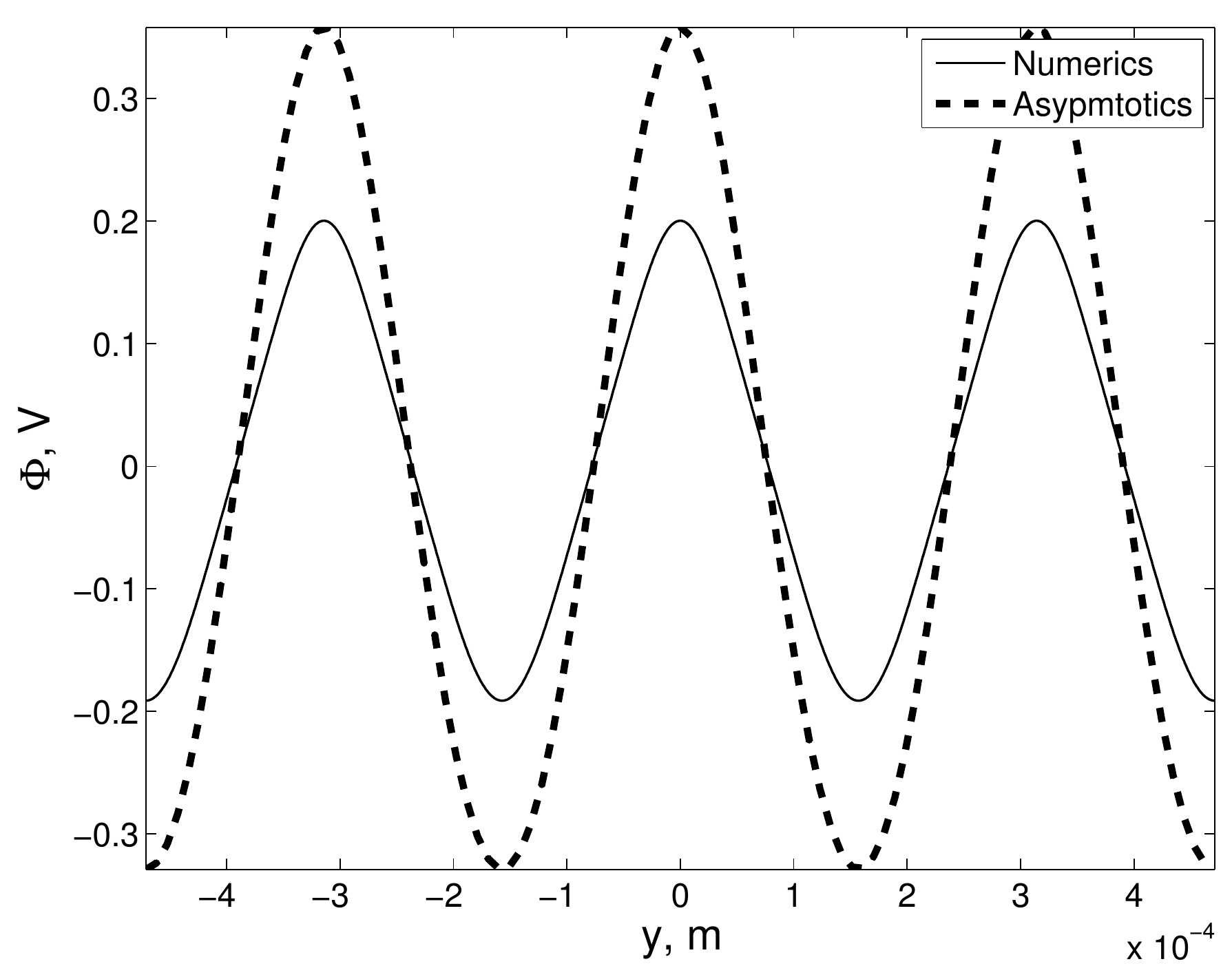}\ \includegraphics[height=2.07in]{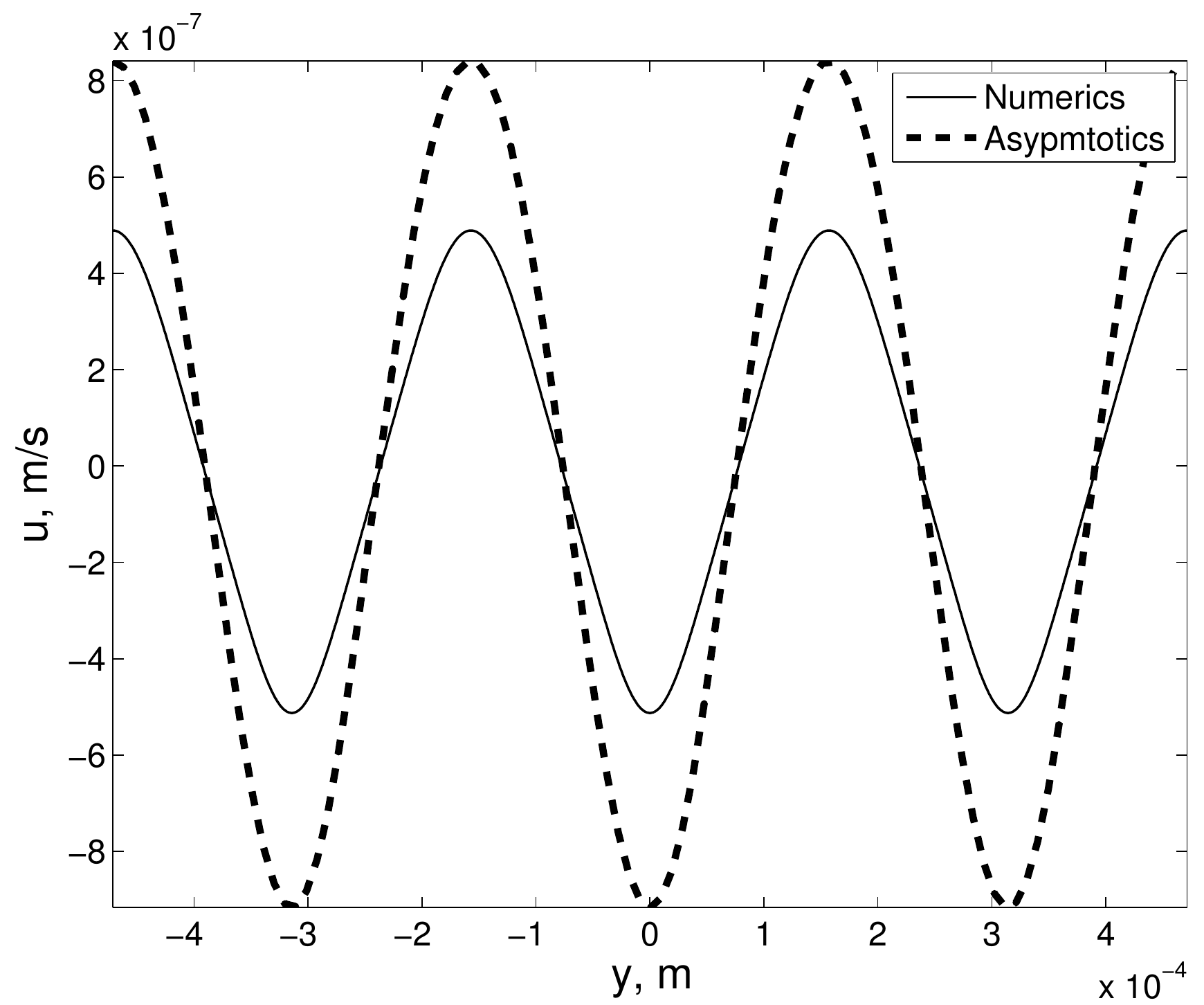}
\caption{Electric potential (left) and flow velocity (right) when $\bar{c}=7\cdot10^{15}\,\mathrm{m}^{-3}$ and the remaining parameters are as in \eqref{eq:dun}.}
\label{fig:phi1}
\end{center}
\end{figure}
Indeed, in Figs. \ref{fig:cp-cm1}-\ref{fig:phi1} we have plotted the numerical and asymptotic
solutions when $\bar{c}=7\cdot10^{15}\,\mathrm{m}^{-3}$ so that $\BN\ll1$. The graphs clearly show significant charge separation and confirm that the asymptotic solution is not valid when $\BN$ is small. {
Further, the dependence of the flow velocity on the field in Fig. \ref{fig:phi-s1} is closer to being linear, as $u_{max}\sim\Phi_0^{1.266}$, indicating that the system approaches saturation. The solutions depicted in Figs. \ref{fig:cp-cm1}-\ref{fig:phi1} can also be obtained via an asymptotic singular perturbation procedure for \eqref{eq:r36.1}
using an appropriate small parameter. This analysis is beyond the scope of this paper and will be presented elsewhere.
\begin{figure}[htbp]
\begin{center}
\includegraphics[height=2in]{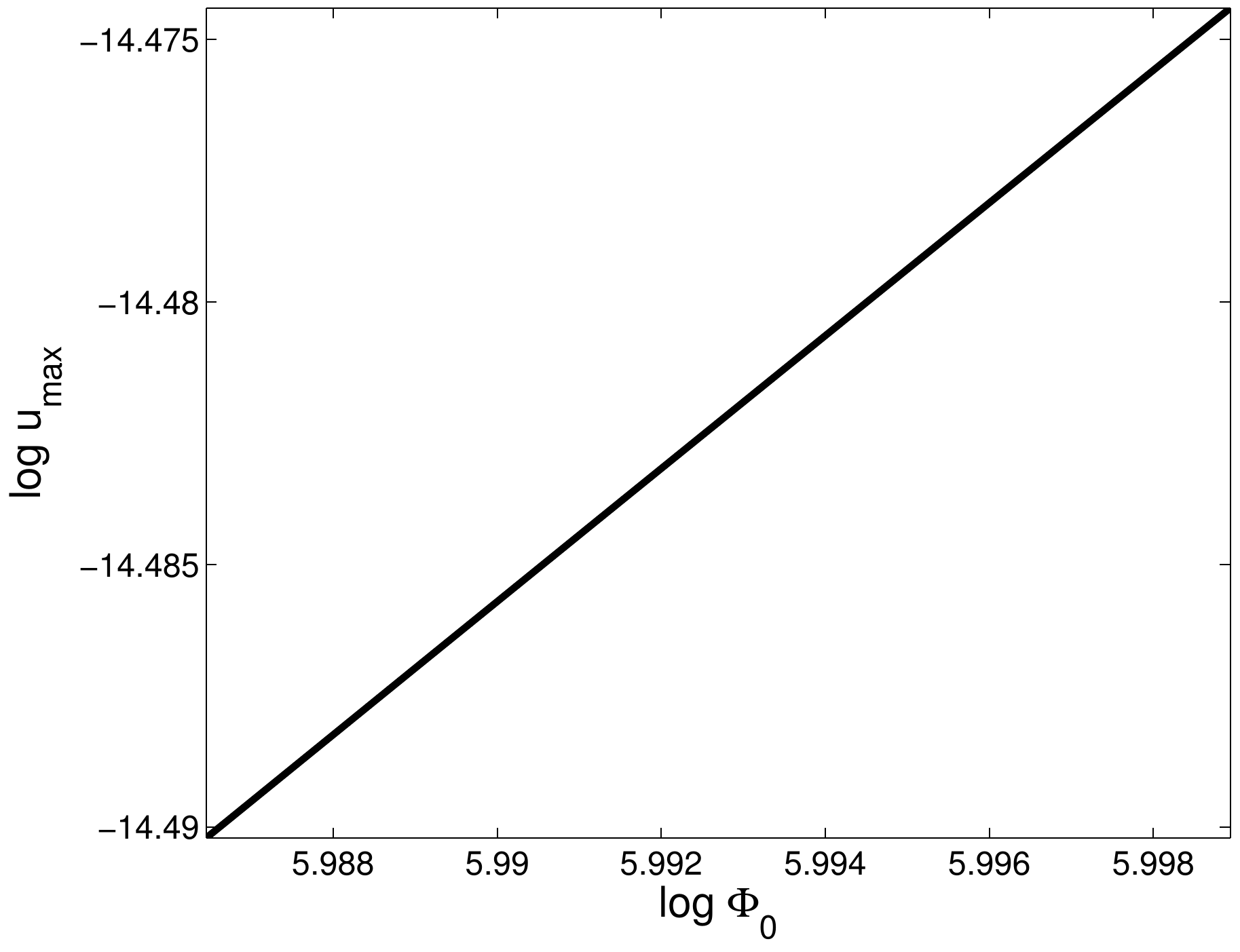}
\caption{Maximum flow velocity as a function of $\Phi_0$ on a logarithmic scale when $\bar{c}=7\cdot10^{15}\,\mathrm{m}^{-3}$ and $\Phi_0\in[398,403]$ V. The slope of the graph is approximately $1.266$.}
\label{fig:phi-s1}
\end{center}
\end{figure}
\section{Acknowledgements}
The authors acknowledge support from National Science Foundation, grant number DMS-DMREF 1435372.

\bibliographystyle{ieeetr}
\bibliography{colloids}  
\end{document}